\documentclass[runningheads]{llncs}
\usepackage[T1]{fontenc}
\usepackage{graphicx}
\usepackage{color}
\usepackage[margin=1in]{geometry}
\usepackage{amsmath,amsfonts,graphicx,amssymb,bm,amsthm}
\usepackage{appendix}
\usepackage{hyperref}
\usepackage{booktabs}
\usepackage{multirow}
\usepackage{makecell}
\usepackage[ruled, vlined]{algorithm2e}
\usepackage{natbib}
\usepackage{cleveref}
\crefname{assumption}{Assumption}{Assumption}
\crefname{algorithm}{Algorithm}{Algorithm}

\def \vtaui {{v^\tau_i}}
\def \vsi {{v^s_i}}

\def \xsttaui {{ x^{*,\tau}_i }}

\def \xtaui {{ x^\tau_i }}

\def \xsi {{ x^s_i }}

\def \vmax {{ v_{\mathrm{max}}}}
\def \vimax {{ v_{i,\mathrm{max}}}}

\def \itau {{ i^\tau }}
\def \utaui {{ u^\tau_i }}

\def \Utaumi {{U^{\tau - 1}_i}}

\def \vitau {{ v_i^\tau }}

\def \ptau {{ p ^ \tau }}

\def \sumiton {{ \sum_{i=1}^n }}

\def \sumtau {{\sum_{\tau=1}^{T}}}

\def \sumiton {{ \sum_{i=1}^n }}
\def \sumstaui {{ \sum_{s = 1}^{\tau - 1}  }}
\def\##1\#{\begin{align}#1\end{align}}
\def\$#1\${\begin{align*}#1\end{align*}}

\begin{document}
\title{
Greedy-Based Online Fair Allocation with Adversarial Input: Enabling Best-of-Many-Worlds Guarantees
}

\author{Zongjun Yang\inst{1}\and
Luofeng Liao\inst{2} \and
Christian Kroer\inst{2}}
 \authorrunning{ Yang et al.}
%
\institute{Peking University \and Columbia University\\
\email{allenyzj@stu.pku.edu.cn, \{ll3530, christian.kroer\}@columbia.edu}}
\maketitle              
\begin{abstract}
We study an online allocation problem with sequentially arriving items and adversarially chosen agent values, with the goal of balancing fairness and efficiency. 
Our goal is to study the performance of algorithms that achieve strong guarantees under other input models such as stochastic inputs, in order to achieve robust guarantees against a variety of inputs.
To that end, we study the 
PACE (Pacing According to Current Estimated utility) algorithm, an existing algorithm designed for stochastic input. We show that in the equal-budgets case, PACE is equivalent to the integral greedy algorithm.  We go on to show that with natural restrictions on the adversarial input model, both integral greedy allocation and PACE have asymptotically bounded multiplicative envy as well as competitive ratio for Nash welfare, with the multiplicative factors either constant or with optimal order dependence on the number of agents. This completes a ``best-of-many-worlds'' guarantee for PACE, since past work showed that PACE achieves guarantees for stationary and stochastic-but-non-stationary input models.
\end{abstract}
\section{Introduction}
We study \emph{online fair allocation}, where items arrive sequentially in $T$ rounds, and we need to distribute them among a set of $n$ agents with heterogeneous preferences, with the goal of achieving both fairness and efficiency properties. In each round, we observe each agent's value for the item and make an irrevocable allocation. Frequently in the literature on this problem, the items are assumed to be divisible, and agents' utilities are linear and additive. We also consider linear and additive utilities, but indivisible items.

Fair allocation problem in the offline setting has been well-studied. A classical objective to optimize is the \textit{Nash welfare} (NW), defined as the geometric mean of the agents' utilities. Maximizing Nash welfare provides a balance between efficiency and fairness due to the multiplicative nature of the objective. For divisible items, an offline optimal allocation can be computed via solving the Eisenberg-Gale (EG) convex program ~\citep{eisenberg1959consensus}. The solution enjoys both envy-freeness and proportionality, which are important measures of fairness. For indivisible items, finding an allocation that maximizes NW is an APX-hard problem even in the offline setting~\citep{moulin2004fair,lee2017apx}, though constant-factor approximation algorithms are known~\citep{cole2018approximating,cole2017convex,mcglaughlin2020improving}. 

In the online setting, \citet{gao2021online} provides a simple allocation algorithm called PACE (Pace According to Current Estimated utility), which generates asymptotically fair and efficient allocations when items are drawn in an i.i.d. manner. PACE gives each agent a per-round budget of faux currency, and simulates a first-price auction in each round. The fair allocation is achieved by having each agent shade their bid with a \textit{pacing multiplier}, which is a projection of their current estimated inverse bang-per-buck to a fixed interval. \citet{liao2022nonstationary} extend these results to non-stationary inputs, where the distribution of items may change over time. They show that in this case, PACE still achieves asymptotic fairness and efficiency guarantees, up to linear error terms from the amount of non-stationarity.

Yet in many real-world scenarios, we cannot expect items to be drawn in a stochastic manner, even if from non-stationary distributions. This motivates the investigation of algorithms with competitive ratio guarantee for adversarial settings. To fit arbitrary inputs, including extreme ones, some algorithms adopt ``conservative'' designs for fairness, such as allocating half of each item purely equally ~\citep{banerjee2022online}. Although this helps to provide worst-case guarantees, it damages the efficiency in the average case, which may not be acceptable in some practical applications.
Moreover, this conservative allocation requires each item to be divisible, or at least for random allocation to be acceptable.

This motivates us to move in another direction: Instead of developing algorithms to fit extreme adversarial inputs, we seek to find worst-case guarantees for existing algorithms that are designed for stochastic inputs. In particular, we focus on the performance of the algorithm PACE ~\citep{gao2021online,liao2022nonstationary}, and explore the question:

$$\textit{How does PACE perform under adversarial input?}$$

Our first contribution is to show that, in the case where all agents have the same weight (or \emph{budget} in market equilibrium terminology), PACE is equivalent to the \textit{integral greedy algorithm}, assuming no projection of the pacing multiplier. 
Due to this equivalence, we start by studying the integral greedy allocation method. Our results for that method are of independent interest, as it is a natural allocation method.

Although both integral greedy allocation and PACE have infinite envy and $\Omega(T)$ competitive ratio when inputs are completely arbitrary, we notice that such pessimistic results only occur under extreme inputs where the ratio between the largest and smallest non-zero values for an agent differ by an exponential factor. We show that, once we rule out such extreme instances by introducing mild assumptions (such as a constant ratio between nonzero values observed by an agent) both algorithms converge with approximate envy-freeness and bounded competitive ratio as $T$ increases to infinity.
The upper bounds are either constant, or in near-optimal order of $n$, see \Cref{table-summary}. Combined with existing results under stationary ~\citep{gao2021online} and non-stationary ~\citep{liao2022nonstationary} input models, this establishes a ``best-of-many-worlds'' guarantee for PACE: it is the first online algorithm that simultaneously guarantees asymptotic fairness and efficiency guarantees under stochastic, stochastic but nonstationary, and adversarial inputs. As such, we believe our results show that PACE is a natural and robust algorithm for online fair allocation in real-world settings, since it achieves strong guarantees under many different utility models, and is thus likely to perform well on a variety of real-world inputs.

\begin{table*}[t]
    \centering
    \begin{tabular}{ccccc}
    \toprule
         \textbf{Algorithm} & \textbf{Assumptions} & \textbf{Measure} & \textbf{Upper-Bound} ($T\rightarrow \infty$) & \textbf{Theorem} \\
    \midrule
        \multirow{3}{*}{\makecell{Integral \\ Greedy}} & \multirow{2}{*}{ Assumption \ref{assumption-greedy}} & multiplicative envy & $1+2\log\frac{1}{\varepsilon}$ & \Cref{thm-envy}\\
        \cmidrule(r){3-5}
        && competitive ratio \textit{w.r.t.} NW & ${\lambda\cdot(n!)^{1/n+\alpha}, \forall \alpha>0}^*$ & \Cref{thm-cr-upper}\\
        \cmidrule(r){2-5}
        & seed utility $\delta$ & utility ratio with seeds & $O(\log T)$ &\Cref{thm:seeded int greedy}\\
    \midrule    
        \multirow{2}{*}{PACE} & Assumption \ref{assumption-pace-1} & multiplicative envy & $1+2\log\frac{1}{\varepsilon}$ & \Cref{thm-envy-pace}\\
        \cmidrule(r){2-5}
        &Assumption \ref{assumption-pace-2}& competitive ratio \textit{w.r.t.} NW  & $\left(1+2\log \frac{1}{\varepsilon}\right) \frac{1}{c}$ & \Cref{thm-cr-pace}\\
    \bottomrule
    \end{tabular}
    
    \medskip 
    $^*$ The lower bound for online algorithms is at least $(n!)^{1/n}$ when $n\rightarrow \infty$.
    
    \caption{Summary of Results}
    \label{table-summary}
\end{table*}
\subsection{Related Work}
In this subsection, we review previous works that are most closely related to ours. An extensive review of other related works is provided in \Cref{appendix-related-work}.
\paragraph{The PACE Algorithm}
Our work is a direct generalization of the PACE algorithm (Pace According to Current Estimated utility)~\citep{gao2021online,liao2022nonstationary} to adversarial inputs. We will review PACE in detail in \Cref{section-setup}.

\paragraph{Online Fair Division}
For maximizing Nash welfare in the online setting, the competitive ratio with respect to the offline optimal allocation is trivially $\Omega(n)$ and $\Omega(T)$, when the input is completely adversarial~\citep{banerjee2022online}. These pessimistic results under arbitrary input motivate the introduction of assumptions. \citet{azar2010allocate} adopted an assumption that we also make: that the minimum \textit{nonzero} valuation of each agent is at least $\varepsilon$ times the largest. 
However, their analysis does not remove the dependence on $T$ in their upper bound of the competitive ratio $O\left(\log \frac{nT}{\varepsilon}\right)$, meaning that the ratio is unbounded in the online setting with a fixed number of agents and an unbounded horizon.
\citet{banerjee2022online} assumes extra prior knowledge of the monopolistic utility of each agent, and gives an $O(\log n)$ and $O(\log T)$-competitive algorithm; their algorithm involves allocating half of each item uniformly across agents; since we are interested in algorithms with asymptotic convergence guarantees on non-adversarial inputs, such an approach cannot be used. 
\citet{huang2022online} assumes the input to be $\lambda$-balanced or $\mu$-impartial, where $\lambda$ and $\mu$ characterize the desired properties of the input; their competitive ratio upper bounds are logarithmic in the parameter and $n$. However, the parameters still implicitly depend on the horizon length $T$. 

The above works on online NW maximization consider greedy-style algorithms as we do. In contrast to them, our work focuses on an asymptotic bound that does not depend on $T$ as $T\rightarrow \infty$. Moreover, while their algorithms can only deal with divisible items, our paper adopts integral allocation, and show that integral decision is enough for the convergence of multiplicative envy and competitive ratio \textit{w.r.t.} the optimal continuous allocation, given our assumptions.

\paragraph{Online Allocation with Resource Constraints} We also briefly discuss how this paper differs from existing works in \textit{online resource allocation}, where a sequence of \emph{requests} arrive over time, with each request consisting of a reward and cost function, and at each time step the algorithm must select a decision with the goal of maximizing the sum of rewards while satisfying long-term cost constraints on each resource. In that setting, strong best-of-many-worlds guarantees are known~\citep{balseiro2023best,celli2022best,castiglioni2023online}. In online resource allocation, the optimization objective is separable across timesteps, e.g., in the form of $\sum_{s=1}^T f_s(x)$~\citep{balseiro2023best}. Time-separability is crucial for the regret bounds in these works, as it enables translating dual regret to primal regret through weak duality. However, in online fair allocation with the Nash welfare objective, time-separability no longer holds, since we take the logarithm of the utility over time. Therefore, our results cannot be derived with similar techniques to those papers.
Indeed, the types of competitive-ratio guarantees achieved e.g. by \citet{balseiro2023best} are impossible in the online fair allocation setting, where hard input sequences are known~\citep{banerjee2022online,gao2021online}.
\section{Setup}
\label{section-setup}
\subsection{Online Fair Allocation}

Consider a problem instance with $n$ agents and $T$ items. For $i\in[n]$ and $t\in [T]$, let $v_{i}^t\geq 0$ be agent $i$'s value for a unit of item $t$. The input of our problem is a sequence of agent valuations $\bm{v} = (v_i^t)^{n\times T}$. We assume that for each item $j$ at least one agent values it, i.e. there exists an agent $i$ such that $v_i^t > 0$. Each agent $i\in [n]$ has a non-negative weight $B_i$, which can also be interpreted as a \textit{budget} of faux currency in a Fisher market~\citep{varian1974equity}. 
An allocation $\bm{x} = (x_i^t)^{n\times T}$ distributes each item to an agent, where $x_i^t$ is the amount of item $t$ that is allocated to agent $i$ gets item $t$. We assume each item has a unit supply. An allocation is \textit{feasible} if $\sum_{i\in[n]}x_{i}^t\leq 1$ for each $t$. An allocation is \textit{integral} if $x_i^t \in \{0,1\}$. For a feasible, integral allocation $x$, let $A_i = \{t: x_i^t = 1\}$ be the set of items allocated to agent $i$.

We assume additive, linear utility for all agents. That is to say, $U_i^t = \sum_{s=1}^t x_{i}^s v_i^s$, where $U_i^t$ is the utility agent $i$ derives from the first $t$ items. For a subset of items $A \subseteq [T]$, agent $i$'s total value on the bundle $A$ is denoted as $U_i(A) = \sum_{s\in A} v_i^s$. Agent $i$'s \textit{monopolistic utility} $V_i$ is defined as his total value for all items $V_i = U_i([T]) = \sum_{t=1}^T v_i^t.$

We focus on the online setting where items arrive sequentially, while the set of agents is fixed. An online allocation algorithm is one that makes an irrevocable choice to distribute the item in each round based only on the information of past rounds. Concretely, it maps the history $\mathcal{H}^t = \left\{(v_i^s)_{s=1}^t,(x_i^s)_{s=1}^{t-1}\right\}_{i=1}^n$ to a decision $(x_1^t, \cdots, x_n^t)$ such that $\sum_{i=1}^n x_i^t =1$. 

In this paper, we are interested in \textit{envy} and \textit{Nash welfare} (NW) as measures of fairness. The multiplicative envy of agent $i$ to agent $j$ is defined as the ratio between the utility that agent $i$ would get from the allocation $x_j$ of agent $j$ to the utility of their own allocation $x_i$, adjusted by their respective budget. As a criterion for fairness, it measures the extent to which an agent prefers someone else's bundle to his own.
\[ \mathrm{Envy}_{ij} = \frac{B_i}{B_j}\frac{\sum_{t=1}^T x_j^t v_i^t}{\sum_{t=1}^T x_i^t v_i^t}.\]

Notice that when the allocation is integral, the above definition becomes ${U_i(A_j)}/{U_i(A_i)}$.

The Nash welfare of an allocation is defined as the weighted geometric mean of all agents' utilities:
\[\mathrm{NW} = \prod_{i=1}^n \left(U_i^T\right)^{B_i/\sum B_j} .\]
Maximizing the geometric mean is a well-studied proxy for balancing fairness and efficiency in fair allocation problems. It is also equivalent to maximizing the objective of the Eisenberg-Gale convex program in a Fisher market. For a Nash welfare maximizing allocation $\bm{x}^\star = \{x^{\star, t}_i\}$, let $U_i^\star = \sum_{t=1}^T x^{\star, t}_i v_i^t$ be the utility of agent $i\in [n]$. Notice that the optimal allocation might not be unique. In this paper, we measure the performance in terms of Nash welfare maximization based on the \textit{competitive ratio} of our allocation, which is defined as the supremum of the ratio between the online allocation and an optimal offline NW-maximizing allocation, over the space of all possible inputs. 

While we allow the input sequence to be adversarial, we restrict our attention to a subset of adversarial input sequences, where we use $\mathcal{V}^T$ to denote the space of valid length $T$ input sequences. Concrete assumptions on $\mathcal{V}^T$ will be specified in \Cref{section-greedy} and \Cref{section-pace} before analyzing specific algorithms.  

We consider the asymptotic worst-case envy and competitive ratio (\textit{w.r.t.} Nash welfare) over all possible inputs $\mathcal{V}^T$ when $T\rightarrow \infty$. Concretely, we will investigate

$$ \lim_{T\rightarrow \infty} \sup_{v \in \mathcal{V}^T} \max_{i,j} \mathrm{Envy}_{ij}, \ \lim_{T\rightarrow \infty} \sup_{v \in \mathcal{V}^T} \prod_{i=1}^n \left(\frac{U_i^T}{U_i^\star}\right)^{B_i/\sum B_j}.$$

In our analysis, we will show that with proper assumptions, both measures converge with an asymptotic upper bound which is independent of $T$. Our upper bounds will be either constant or with a near-optimal order dependence on $n$.

We emphasize that both measures adopted in this paper, multiplicative envy, and competitive ratio, are defined as ratios, not differences. This is mainly because of the \textit{scale invariant} property of fair division, which is an important and desirable property for allocation algorithms: if an agent's values for all items are multiplied by a constant factor, the resulting allocation stays the same. The algorithms that we are interested in, together with NW-maximizing allocations, are all scale invariant. Hence, it is more useful to consider multiplicative performance measures, which are also invariant to valuation scaling.

\subsection{Algorithms}
We introduce the two major algorithms that we study in this paper:
the PACE (Pace According to Current Estimated Utility) algorithm~\citep{gao2021online,liao2022nonstationary}, and the integral greedy algorithm.
Moreover, we will discuss the ``greedy-based'' nature of PACE by showing its equivalence to the integral greedy algorithm under certain conditions.

Pseudocode for PACE is shown in Algorithm \ref{alg-pace}. In each round $t$, the agent utilities are revealed. Each agent then places a bid for the item, which is equal to their value for the item multiplied by the current pacing multiplier $\beta_i^t$. The whole item is allocated to the highest bidder, preferring the bidder with the smallest index when a tie occurs. Each agent then observes their realized utility at this round, and updates their current estimated utility. The pacing multiplier is updated to be the weight $B_i$ divided by the estimated utility, then projected to an interval $[a,b]$. 
\begin{algorithm}
\caption{PACE}
\label{alg-pace}

    \SetKwInput{KwInit}{Initialization}
    \KwIn{number of agents $n$, time horizon $T$, truncation parameters $a$ and $b$.}
    \KwInit{$U_i^0 = 0, \beta^1 = 1^n$.}
    \For{$t = 1, \cdots, T$,}{
        Agent $i$ bids $\beta_i^t v_i^t$. 
        
        The whole item $t$ is allocated to the highest bidder, with arbitrary tie-breaking:
        $$i^t := \min\arg \max_{i\in[n]} \beta_i^t v_i^t, \ \  x_i^t = \bm{1}(i =i^t).$$
        Agent $i$ updates his estimated utility
        $$ \Bar{u}_i^t= \frac{1}{t} \cdot x_i^t v_i^t + \frac{t-1}{t}\Bar{u}_i^{t-1}.$$
        Agent $i$ updates the pacing multiplier 
        $$ \beta_i^{t+1} = \prod_{[a, b]}\left(\frac{B_i}{\Bar{u}_i^t}\right)$$
    }
\end{algorithm}
\subsubsection{Performance Under Stochastic Input}
As shown by \citet{gao2021online}, PACE is an instantiation of stochastic unregularized dual averaging ~\citep{xiao2009dual} applied to the dual of the underlying allocation program where the supplies are given by the density of each item. With i.i.d. input, PACE converges to the equilibrium of a potentially infinite-dimensional Fisher market~\citep{gao2023infinite}, which is also closely related to the game-theoretic solution concept of pacing equilibrium~\citep{conitzer2022pacing}. The agent utilities under PACE converge to those associated with the offline NW-maximizing allocation in the mean-square sense. 

\begin{theorem} \label{lemma-stationary} \textnormal{(Theorem 4. in \cite{gao2021online})} 
    Let $u_i^\star:= U_i^{\star}/T$ be agent $i$'s time-averaged utility under the Nash-welfare-maximizing allocation with supplies given by some underlying distribution. When the values in different rounds are i.i.d. chosen from the same distribution, it holds that
\begin{equation}
    \bm{E}\left[ \sum_{t=1}^T(\bar{u}_i^t - u_i^{\star})\right] \leq C \cdot  \frac{\log T}{T},
\end{equation}
where $C$ is a constant independent of $T$.
\end{theorem}
\citet{liao2022nonstationary} generalizes \Cref{lemma-stationary} to non-stationary inputs, which have a stochastic component yet change over time. Particularly, they consider three types of inputs: independent yet adversarially corrupted input, ergodic input, and periodic input, and show that for all three cases $\bm{E} \|u^T - u^\star\| \rightarrow 0$ is still preserved, up to errors due to non-stationarity. 

In the stationary and non-stationary cases, mean-square convergence of time-averaged utility implies an asymptotic competitive ratio of $1$ \textit{w.r.t} Nash welfare. For both cases, \citet{gao2021online,liao2022nonstationary} also provide a theoretical guarantee that PACE is asymptotically envy-free (again up to a non-stationarity error in the non-stationary case). In this paper, we provide bounds on PACE's performance with adversarial inputs with assumptions. Combined with the results from \citet{gao2021online,liao2022nonstationary}, this is the first ``best-of-many-worlds'' guarantee for online fair allocation under stationary, non-stationary, and adversarial inputs. 

While we make an attempt to take an algorithm for stochastic inputs and show its performance on adversarial input, the other direction seems to be difficult. It is hard for some algorithms that are designed for adversarial inputs to achieve optimality in stochastic scenarios, due to ``conservative'' routines that they adopt to deal with extremely bad inputs. For instance, \citet{banerjee2022online} divides half of the resources equally, which can be undesirable with stochastic input, see \Cref{example-stochastic}

\begin{example}
\label{example-stochastic}
    Consider an online scenario with $n$ agents and $n$ types of items $\{\theta_j\}_{j=1}^n$. Agent $n$'s value for a unit of type-$j$ item $v_i(\theta_j) = 1$ if $i = j$, and $v_i(\theta_j) = 0.01$ otherwise. In each round, the item type is drawn i.i.d. from a uniform distribution. In this scenario, 
    \begin{itemize}
        \item The equilibrium of the underlying Fisher market allocates all type-$i$ items to agent $i$. PACE converges to this equilibrium.
        \item The proposed algorithm proposed by \citep{banerjee2022online} allocates at most $\left(\frac{1}{2}+\frac{1}{n}\right)$ fraction of type-$i$ item to agent $i$, which is clearly not optimal.
    \end{itemize}
    
\end{example}

\subsubsection{Greedy Interpretation of PACE.} 
In the standard configuration of online fair allocation, where agents have equal weight $B_1 = \cdots = B_n = 1$, we now show that PACE, when the projection of multipliers is disregarded, can be interpreted as greedily maximizing Nash welfare with integral decisions.

To show this, the following optimization program \eqref{eq-greedy-opt} maximizes NW up to round $t$ greedily, given the history of previous $t-1$ rounds. Its decision is integral.
\begin{equation}
\label{eq-greedy-opt}
    \begin{aligned}
    \max_{x_{i}^t \in \{0,1\}} & \ \ \sum_{i=1}^n B_i \log U_{i}^t \\
    \text{s.t.} & \ \ U_{i}^t = \sum_{s=1}^{t-1} x_{i}^s v_{i}^s + x_{i}^t v_{i}^t ,\ \forall i\in [n] \\
                 & \ \ \sum_{i=1}^n x_{i}^t = 1 \\
    \end{aligned}
\end{equation}
The above program allocates the item to the agent $i^t$ that gives the maximum increment to the objective. 
This is equivalent to the following decision rule:
\begin{equation*}
    i^t \in \arg \max_{i\in [n]} B_i \log \left(1+\frac{v_i^t}{U_i^{t-1}}\right).
\end{equation*}
When agent weights are equal, $i^t$ is the agent $i$ that maximizes ${v_i^t}/{U_i^{t-1}}$; this coincides exactly with the decision of PACE with no projections.

This interpretation motivates us to first consider the behavior of a greedy algorithm (without projecting the multiplier), and then study PACE based on the insights derived for integral greedy. We focus on\textit{ weight-adapted integral greedy algorithm}, or simply integral greedy algorithm, as shown in \Cref{alg-greedy}.
\begin{algorithm}
\caption{Integral Greedy Algorithm (weight-adapted)}
\label{alg-greedy}

    \SetKwInput{KwInit}{Initialization}
    \KwIn{number of agents $n$, time horizon $T$}
    \KwInit{$U^0_i = 0$ for all $i$.}
    \For{$t = 1, \cdots, T$,}{
        Observe agent values for item $t$, and allocates the whole item to agent $i^t$:
        $$i^t := \min \bigg(\mathrm{arg} \max_{i\in[n]} \frac{B_i v_i^t}{U_i^{t-1}} \bigg), \ \  x_i^t = \bm{1}(i =i^t).$$
        Agent $i$ updates his current utility
        $$ U_i^t = U_i^{t-1}+ x_i^t v_i^t.$$
    }
\end{algorithm}

Algorithm \ref{alg-greedy} is equivalent to greedily maximizing the NW objective with integral decisions when agent weights are equal. However, with unequal weights the equivalence breaks down, so does its equivalence between PACE. For this consideration, in the following discussions, we assume equal weights $B_1 = \cdots =B_n$, which is a standard setting in the Nash welfare maximization and online fair allocation literature.
We will also remark on part of our results that can be generalized to unequal weights. 
Note that the PACE algorithm itself extends to unequal weights, and the stationary and nonstationary results extend as well.

With equal weight assumption, PACE can be regarded as greedy-based decision using $\Hat{U}_i^t$ as estimated utility. $\Hat{U}_i^t$ is a projected utility, where $\ell$ and $r$ are reinterpreted projection bounds:
\begin{equation*}
    i^t := \min\mathrm{arg} \max_{i\in[n]} \frac{v_i^t}{\Hat{U}_i^{t-1}}, \ \Hat{U}_{i}^{t} = \prod_{[\ell t, r t]} U_i^t.
\end{equation*}
\section{Analysis of the Integral Greedy Algorithm}
\label{section-greedy}
This section is devoted to show that integral greedy algorithm achieves convergence of multiplicative envy and competitive ratio, under reasonable assumptions. The missing proofs in this section are deferred to \Cref{appendix-greedy}.
\subsection{Assumptions on the Input}
We begin with introducing the assumptions on the input space, as well as the necessity of doing so. We focus on input space $\mathcal{V}_{\varepsilon}^T$, which is parametrized by $\varepsilon\in (0,1]$ and defined as follows:

\begin{assumption}
\label{assumption-greedy}
    For each integer $T$, $\mathcal{V}_{\varepsilon}^{T}$ is the set of inputs which satisfy the following requirements:
    \begin{itemize}
        \item The monopolistic utility of each agent is infinite: $V_i = \infty \ (T\rightarrow \infty).$ 
        \item For each $i\in[n]$ and $t\in[T]$, $v_{i}^t \in \{0\}\cup[\varepsilon, 1]$.
    \end{itemize}
\end{assumption}

The first requirement helps to avoid allocating nothing to some agent, since the decision is integral. We further require the number of nonzero valuations of each agent is not upper-bounded by a constant, so it is meaningful to consider the asymptotic sense in $T$. For the second requirement, we note that it is equivalent to assuming a constant bound on the ratio of the minimum nonzero item value and the maximum:
$$ \frac{\min_t\{v_i^t: v_i^t>0\}}{\max_t\{v_i^t\}} \geq \varepsilon, \ \forall i\in[n].$$
This equivalence is due to the scale-invariant property of integral greedy allocation, as well as the NW maximizing allocation.

If we do not include any extra assumptions on the input space, the worst-case envy and the competitive ratio is infinite, which makes the analysis of the integral greedy algorithm trivial and uninteresting. 
\begin{lemma}
\label{lemma-general}
When the values are arbitrarily chosen from $[0,1]$, it holds that
\begin{enumerate}
    \item \textnormal{~\citep{banerjee2022online}} Any online allocation algorithm has $\Omega(T)$ competitive ratio with respective to NW.
    \item The integral greedy algorithm has $\Omega(T)$ worst-case multiplicative envy.
\end{enumerate}
even when the first requirement in \Cref{assumption-greedy} is satisfied.
\end{lemma}
\begin{proof} 
We construct a sequence of instances where $\mathrm{Envy}_{12} = \Omega(T)$. Given horizon length $T$, fix agent $1$'s valuation to $1$ in all rounds. For agent $2$, we set $v_2^t = a^{t-T} \leq 1$ for some $a>2$. In the integral greedy allocation, agent $1$ receives only one item and has total utility $1$. This proves the lower bound of multiplicative envy.
\end{proof}

The above hard instance features exponential growth in agent $2$'s valuation. Hard instances with similar spirit for the continuous problem have previously been given by \cite{banerjee2022online}. The vulnerability of allocation algorithms under such instances can be explained by their non-anticipating nature: it does not know the future, so it fails to recognize that agent $2$'s valuation up to current rounds is only a negligible fraction of the entire horizon. More generally, it is difficult for online algorithms to distinguish agents that are hard to satisfy in the future, with those who are easily satisfied. 

However, we note that such adversarial instances are arguably not ``natural.'' For allocation problem in a real-world market, items are usually similar in nature, e.g., they are all food or ad slots. It is unlikely for an agent to have exponentially diverging nonzero values on these items. 
By requiring the ratio of minimal and maximal nonzero value to be bounded by $\varepsilon$, \cref{assumption-greedy} rules out extreme cases where values from a single agent diverge drastically. In the rest of this section, we will show that once the above assumptions are introduced, both multiplicative envy and competitive ratio of the greedy algorithm are independent of $T$ asymptotically, i.e., converge to a constant (which depends only on $n$ and $\varepsilon$).

\subsection{Envy Analysis for Greedy}
In this subsection we analyze the worst-case envy of agents in integral greedy allocation. We observe that, the envy between any pair of agents can be reduced to the case $n=2$ by the following lemma, which characterized an ``induction'' structure of the integral greedy allocation. 
\begin{lemma}[Inductive structure of greedy allocation.]
    \label{lemma-induction-property}
    For any $n$-agent instance $\bm v$ and agent subset $I \subseteq [n]$, define a new instance $\bm{v}|_{I}$ obtained by transforming $\bm v$ as:
    \begin{itemize}
        \item Remove all agents that are not in $I$ (by setting their values to $0$ on all items).
        \item Remove all items that are not in $\bigcup_{i\in I} A_i$ (by setting all agents' value to $0$ on them).
    \end{itemize}
    Then, the resulting allocation is the same when the integral greedy algorithm is run on $\bm{v}$ and $\bm{v}|_{I}$. 
\end{lemma}

We show that the multiplicative envy of the integral greedy algorithm is upper bounded by $1$ plus a logarithmic term in $1/\varepsilon$, which also generalizes to unequal weights.
\begin{theorem}[Upper Bound for Multiplicative Envy]
\label{thm-envy}
Even with unequal agent weights, for inputs $\bm{v}$ satisfying Assumption \ref{assumption-greedy} and any $i, j$, 
\begin{equation*}
\sup_{v\in \mathcal{V}_{\varepsilon}^T}\mathrm{Envy}_{ij} 
\leq 1 + 2\log \frac{1}{\varepsilon} + O\left(\frac{1}{T}\right).
\end{equation*}

\end{theorem}
\begin{paragraph}{Proof Sketch.}
    Due to \Cref{lemma-induction-property} and symmetry, it suffices to consider $\mathrm{Envy}_{21}$ in $2$-agent inputs. In the proof sketch, we assume $B_1 = B_2 = 1$ for simplicity. We transform the input by:
    \begin{enumerate}
        \item Set agent $1$'s valuation for all items in $A_2$ to $0$. 
        \item Move all items in $A_2$ to the beginning of the input sequence, and all items in $A_1$ to the end of the input sequence (preserving order).
    \end{enumerate}
    One can show that the agent's utilities under the budget-adapted greedy algorithm are invariant to this transformation. Therefore, it suffices to consider only the transformed inputs, where agent $2$ receives his entire share only in beginning $R$ rounds. 
    
    To find the worst-case envy for transformed inputs, a question from an adversarial point of view will be: given $(U_1^R, U_2^R) = (0, U)$, how can we design a value sequence for the coming rounds, such that agent $2$'s total valuation on items over $S$ rounds is maximized, while ensuring that nothing is allocated to agent $2$? This can be characterized by an optimization program:  
\begin{equation}
\label{eq-optimization-canonical}
    \begin{aligned}
    \max_{v_1^t, v_2^t} & \ \ \frac{1}{ U}\sum_{t=1}^{\infty} v_2^t \\
    \text{s.t.} & \ \ v_1^t, v_2^t \in \{0\}\cup[\varepsilon, 1], \ \ &\forall t \geq 1\\
                & \ \ \frac{v_2^t}{v_1^t} \leq \frac{U}{U_1^{t-1}}, \ \ &\forall t \geq 1\\
                & \ \ U_1^t \geq \sum_{s=1}^{t} v_{1}^s, \ \ &\forall t \geq 1 \\
    \end{aligned}
\end{equation}
Notice that in \eqref{eq-optimization-canonical} we re-index the rounds by starting with index $1$ at round $R+1$.
We call \eqref{eq-optimization-canonical} a \textit{canonical optimization program} for multiplicative envy maximization, parametrized by $U$. We observe that $v_2^t/v_1^t$ can be upper bounded by $q(U_1^{t-1})$, defined as
\begin{equation*}
    q(U_1^{t-1}) = \begin{cases}
        \min\left\{\frac{U}{U_1^{t-1}}, \frac{1}{\varepsilon} \right\}, & 0\leq U_1^{t-1} \leq \frac{U}{\varepsilon} \\
        0, &U_1^{t-1} > \frac{U}{\varepsilon} 
    \end{cases}.
\end{equation*}

    We can then give an upper-bound of the objective of the canonical optimization program \eqref{eq-optimization-canonical},
\begin{equation}
\label{eq-bound-with-q}
    \frac{1}{U}\sum_{t=1}^\infty v_2^t = \frac{1}{U}\sum_{t=1}^\infty v_1^t \cdot \frac{v_2^t}{v_1^t} \leq \frac{1}{U}\sum_{t=1}^\infty v_1^t \cdot q(U_1^{t-1}).
\end{equation}
    Each increment $v_1^t \leq 1$ is small when $U\rightarrow \infty$. One can show that the right hand side of \eqref{eq-bound-with-q} converges to a definite integral asymptotically as $U\rightarrow \infty$:
    $$\frac{1}{U}\sum_{t=1}^\infty v_1^t \cdot q(U_1^{t-1}) \rightarrow \frac{1}{U} \int_0^{U/\varepsilon} q(U_1) \mathrm{d}U_1 
    = 1 +2 \log \frac{1}{\varepsilon}.$$
    The convergence rate on the order of $O(1/T)$ is then achieved by more carefully calculating the above upper bound. 
    $\hfill \square$
\end{paragraph}    

    The analysis of the canonical optimization program \eqref{eq-optimization-canonical} also gives an upper bound on the number of items that agent $2$ ``envies of'' agent $1$, concretely, items in $A_1$ with nonzero value from agent $2$. \Cref{lemma-number-of-items} will be useful in \Cref{section-pace}. 

\begin{lemma}
\label{lemma-number-of-items}
    Let $C_i:=\{t: v_i^t>0\}$ be the set of items on which agent $i$ has non-zero value. When agents have equal weights, for $\bm{v}\in \mathcal{V}_{\varepsilon}^T$ and any $i,j \in [n]$, 
    $$ \lim_{T\rightarrow \infty} \sup_{v\in V_{\varepsilon}^T} \frac{|C_i\cap A_j|}{U_i(A_i)} \leq \frac{1}{\varepsilon}\left(1+\log \frac{1}{\varepsilon}\right).$$
    
\end{lemma}

    Next, we complement \Cref{thm-envy} with a lower bound showing that the bound $1+2\log 1/\varepsilon$ on envy is tight for the integral greedy allocation. 
    
\begin{theorem}[Lower Bound for Multiplicative Envy]
\label{thm-envy-lower}
For input $\bm{v}\in\mathcal{V}_{\varepsilon}^T$, we have for any $i, j\in [n]$, 
$$  \lim_{T\rightarrow \infty} \sup_{v \in \mathcal{V}_{\varepsilon}^T} \frac{B_i}{B_j} \cdot  \frac{U_i(A_j)}{U_i(A_i)} \geq 1 + 2\log \frac{1}{\varepsilon}.$$
\end{theorem}

\subsection{Nash Welfare Analysis for Greedy}
We give an upper bound on the asymptotic competitive ratio \textit{w.r.t.} Nash welfare for the integral greedy algorithm.

\begin{theorem}[Upper Bound for Competitive Ratio]
\label{thm-cr-upper}
    For input space $\mathcal{V}_{\varepsilon}^T$ and any given $\alpha>0$ there exists a constant $\lambda>0$ (independent of $n$ and $T$), such that
$$  \lim_{T\rightarrow \infty} \sup_{v \in \mathcal{V}^T_\varepsilon} \left(\frac{\prod_{i=1}^n U_i(A_i)}{\prod_{i=1}^n U_i(A_i^{\star})} \right)^{1/n} \leq \lambda\cdot\left(n!\right)^{\frac{1+\alpha}{n}}.$$

For $\varepsilon = 1$, the above holds with $\lambda = 1, \alpha  = 0$.
\end{theorem}
\begin{paragraph}{Proof Sketch} Assume without loss of generality that $U_1(A_1)\leq\cdots \leq U_n(A_n)$. The main idea of the proof is to show that $U_i^{\star}/U_i$ is bounded by $(n-i+1)\cdot i^{\alpha}$ asymptotically. Suppose this is not true, we show that $x_i^{\star}$ will include large proportion of $x_j$ ($j<i$), which will lead to contradiction with the optimality of $\bm{x}^\star$.
    
\end{paragraph}

Although \citet{azar2010allocate} gives an $O(\log(nT/\varepsilon))$ algorithm, we show that $(n!)^{1/n}$ factor is inevitable if one aims to remove the dependence on $T$. Hence, integral greedy algorithm is near-optimal in terms of $n$.
\begin{theorem}[Lower Bound for Competitive Ratio]
    \label{thm-opt}
    Even in the case $\varepsilon = 1$, for any feasible, deterministic online algorithm we have
    $$\left(\frac{\prod_{i=1}^n U_i(A_i)}{\prod_{i=1}^n U_i^\star} \right)^{1/n} \geq (n!)^{1/n}.$$
\end{theorem}
\paragraph{Proof Sketch.}
We sketch the construction of an adaptive adversary, who attempts to make low-utility agents hard to satisfy in the future. Divide the horizon into $n$ phases, each with length $T_i$, satisfying $T_{i}/T_{i-1} \rightarrow \infty.$ The adversary maintains a set of ``active agents'', initially containing all agents. In each round, only currently active agents see nonzero values. At the end of each phase, the agent in the active set who has lowest utility is eliminated. This results in a competitive ratio of $(n!)^{1/n}$, see the detailed proof in \Cref{proof-thm-opt}.

\subsection{Nash Welfare Analysis without \Cref{assumption-greedy}}
As an extension for our analysis on the integral greedy algorithm, we show how it can be adapted when \Cref{assumption-greedy} does not hold.
Despite the $\Omega(T)$ lower bound, we show that, when each agent begins with a \textit{seed utility} $\delta$, the competitive ratio of integral greedy algorithm is of order $O(\log T)$. The performance is measured \textit{w.r.t.} to the ratio of seeded welfare. 

The \textit{seeded} integral greedy algorithm is identical to \Cref{alg-greedy}, except that all agents are given an initial \emph{seed utility} $\delta$, which is taken into account when deciding the winning agent: 
\begin{align}
            i^t := \min \bigg(\mathrm{arg} \max_{i\in[n]} \frac{v_i^t}{\delta + U_i^{t-1}} \bigg), \ \  x_i^t = \bm{1}(i =i^t).
            \label{eq:seeded_int_greedy}
\end{align}

The full pseudocode for the seeded algorithm is presented in \Cref{alg-greedy-seeded}. 

\begin{algorithm}
\caption{Seeded Integral Greedy Algorithm (Equal Weights)}
\label{alg-greedy-seeded}
    \SetKwInput{KwInit}{Initialization}
    \KwIn{number of agents $n$, time horizon $T$, seed utility $\delta$}
    \KwInit{$U^0_i = 0$ for all $i$.}
    \For{$t = 1, \cdots, T$,}{
        Observe agent values for item $t$, and allocates the whole item to agent $i^t$:
        \begin{align}
            i^t := \min \bigg(\mathrm{arg} \max_{i\in[n]} \frac{v_i^t}{\delta + U_i^{t-1}} \bigg), \ \  x_i^t = \bm{1}(i =i^t).
        \end{align}
        Agent $i$ updates current utility
        $$ U_i^t = U_i^{t-1}+ x_i^t v_i^t.$$
    }
    \KwOut{Allocations $\{x_i^t\}_{i,t}$}
\end{algorithm}
For the seeded algorithm, we study the criterion $R_{\delta}(\bm{v})$ which is also defined \textit{w.r.t.} the seeded utility:
\begin{align*}
    R_\delta(v) = \sup_{\widetilde{U}} \Bigg\{ \frac{1}{n} \sum_{i=1}^n \frac{\widetilde{U}_i + \delta}{{U}_i + \delta}\Bigg\} \;,
\end{align*}
where the supremum runs over all feasible hindsight allocations, with resulting utility $\widetilde{U}_1, \cdots, \widetilde{U}_n.$
Notice that by the AM-GM inequality, $R_{\delta}$ is also competitive \textit{w.r.t.} the geometric mean.
However, due to the presence of the seed utility $\delta$, it is not directly comparable to the criterion $(\prod_{i=1}^n  {U^\star_i }/{U_i })^{1/n}$. 

\begin{theorem}[Upper Bound of $R_{\delta}$ for Seeded Integral Greedy]
\label{thm:seeded int greedy}
  Run the seeded integral greedy algorithm with seed utility $\delta$.
  For any input $\bm{v}$ satisfying $v_i^t \in [0,1]$,
\begin{align}\label{eq:thm seeded simplified}
       R_\delta(v) \leq 3 + \frac{4}{\delta}  + 2\log\Big(1+\frac{1}{\delta}\Big) + 2 \log T \;.
\end{align}
\end{theorem}

The ratio is of order $\log T$, similar to the result given by \citet{azar2010allocate}. Unlike our previous results, here the ratio is unbounded as $T \to \infty$ since we are considering a broader range of inputs beyond \Cref{assumption-greedy}.
Also, there is a trade-off in choosing seed utility $\delta$. Although a larger $\delta$ brings a better bound in \eqref{eq:thm seeded simplified}, it is less able to tell us how the algorithm compares to the offline optimal, since $R_{\delta}(v) \rightarrow 1$ as $\delta$ grows large. 

\section{Analysis of PACE}
\label{section-pace}
We continue to analyzing the performance of PACE. Missing proofs in this section can be found in \Cref{appendix-pace}.
\subsection{Assumptions on the Input}

To analyze PACE in the adversarial setting, it is necessary to adopt more assumptions than for the integral greedy algorithm. This is because PACE projects the average utility up to round $t$ to $[\ell, r]$ in its decision, or equivalently, projects the total utility in round $t$ to $[\ell t, r t]$. In the stationary and non-stationary setting of \citet{liao2022nonstationary}, the expected utility of each agent grows uniformly with time; in that case the projection helps PACE achieve theoretical guarantees on its performance. However, in the adversarial setting, the input may vary drastically over time, which makes the projection operation more problematic. In the worst case, adversarial input might cause extreme behavior of the projection, such as not allocating any item to a certain agent. 

In their most generic form, our assumptions require that each agent achieves infinite utility as $T\rightarrow \infty$ and bounded non-zero valuation ratios; see Assumption \ref{assumption-pace-1}. Notice that unbounded utility is necessary if we hope to derive meaningful convergence guarantees for worst-case envy and Nash welfare.
\begin{assumption}
\label{assumption-pace-1}
    For each integer $T>0$, ${\mathcal{V}}_{\varepsilon}^{T}(\ell, r)$ is the set of inputs which satisfy:
    \begin{itemize}
        \item The utility of each agent under PACE is infinite when the projection bounds are set to be $[\ell, r]$.
        \item For each $i\in[n]$ and $t\in[T]$, $v_{i}^t \in \{0\}\cup[\varepsilon, 1]$.
    \end{itemize}
\end{assumption}
Since \cref{assumption-pace-1} is potentially hard to verify, we next identify sufficient conditions under which PACE guarantees infinite utility for each agent. 
\begin{assumption}
\label{assumption-pace-2}
    For integer $T>0$, and $c\in(0,1]$, ${\mathcal{V}}_{\varepsilon,c}^{T}$ is the set of inputs which satisfy:
    \begin{itemize}
        \item For each $i\in[n]$, $V_i \geq cT$, that is, the monopolistic utility of each agent under PACE is $\Theta(T)$.
        \item For each $i\in[n]$ and $t\in[T]$, $v_{i}^t \in \{0\}\cup[\varepsilon, 1]$.
    \end{itemize}
\end{assumption}
Assumption \ref{assumption-pace-2} strengthens Assumption \ref{assumption-greedy} by requiring the monopolistic utility of each agent to be linear in $T$. Intuitively, this matches the linear projection bounds $[\ell t, rt]$ on utility. We will show in \Cref{lemma-infinite-utility} that, with appropriate initialization of the projection bounds, Assumption \ref{assumption-pace-2} with any $c>0$ leads to infinite agent utilities, and thus implies Assumption \ref{assumption-pace-1} for some $\ell$ and $r$.

Conversely, if an agent's inputs have $o(T)$ agent monopolistic utility, the following example shows that problems may occur, see \Cref{example-sublinear}.
\begin{example}
\label{example-sublinear}
Cases with $\sum_{t=1}^T \bm{1}\{v_i^t>0\} = f(T) = o(T)$ for $i=1, 2$ in an $n$-agent case ($n>2$). For any $\ell \in (0,\varepsilon)$, there exists a sufficiently large $T_1$ such that $f(T_1)<\frac{\ell}{\varepsilon}T_1$. We can construct an instance where agents $1$ and $2$ have value zero for the first $\left(1-\frac{\ell}{\varepsilon}\right)T_1$ items and value $\varepsilon$ for the remaining items. One can check that both utilities never reach $\ell t$, and thus $\Hat{U}_1^t = \Hat{U}_2^t = \ell t$ always. Under lexicographical tie-breaking, agent $1$ will receive all items. Agent $2$ will receive nothing at all.
\end{example}

We remark that in order to initialize $\ell$ and $r$, PACE requires knowing a lower bound on $\varepsilon$. 
In contrast, the integral greedy algorithm does not require such knowledge.
A minimal requirement is to set $\ell \leq \varepsilon$:
\begin{example}
\label{example-projection}
Case with $\ell > \varepsilon$: Consider an input where every agent has value $\varepsilon$ in each round. The projection is then activated in every round for every agent, and the algorithm allocates all items to the first agent.
\end{example}

\subsection{Envy Analysis for PACE}
Next we show that the PACE algorithm, as long as it is appropriately initialized and  \cref{assumption-pace-1} holds, achieves $1+2\log 1/\varepsilon$ multiplicative envy bound asymptotically.
Our proof is performed by reducing the adversarial envy-maximizing problem in PACE to the canonical program \eqref{eq-optimization-canonical} by showing that PACE with our assumptions lead to a weakly harder problem than \eqref{eq-optimization-canonical} for the adversary.

\begin{theorem}[Upper Bound of Multiplicative Envy for PACE]
\label{thm-envy-pace}
Under Assumption \ref{assumption-pace-1} with $$ \ell < \frac{\varepsilon^2}{\varepsilon+1+(n-1)\left(1+\log(1/\varepsilon)\right)},\ r = 1,$$
PACE achieves
$$  \lim_{T\rightarrow \infty} \sup_{\bm v \in {\mathcal{V}}_{\varepsilon}^{T}(\ell, r)} \mathrm{Envy}_{ij} \leq 1 + 2\log \frac{1}{\varepsilon}, \ \ \forall i, j \in[n].$$
\end{theorem}
\paragraph{Proof.}
 Due to symmetry, we only focus on the envy of agent $1$. Similar to the proof of \Cref{thm-envy}, we first transform the input $\bm v$ into $\bm v^\prime$, such that the allocation of PACE is preserved. Notice that in this proof we use a different transformation, since the problem cannot be reduced to pairwise 2-agent cases. It is no longer obvious that the transformation preserves the PACE allocation; we will prove this later. The transformation is as follows:
\begin{enumerate}
        \item For agent $j \in \{2, 3, \cdots, n\}$, set agent $j$'s value to $0$ for all items that are not allocated to $j$.
        \item Move all items on which agent $1$ has value $0$ to the end of the input sequence.
        \item Move all items in $A_1$ to the beginning of the input sequence.
\end{enumerate}
Allocation for items in Step 2. is clearly preserved under PACE and does not affect the envy of agent 1, so we can assume such items do not exist without loss of generality.

Denote $S = |A_1|.$ 
For item $s$ in the original input $\bm v$, let $s^\prime > s$ be the round it appears in $\bm{v}^\prime$. We define
$
    t^* := \frac{\varepsilon S}{\varepsilon-\ell}\left(1+\frac{1}{\varepsilon}(n-1)\left(1+\log\frac{1}{\varepsilon}\right)\right).$
The proof is divided into three parts.

For the first part, \Cref{lemma-auxiliary-1} shows that in the first $t^*$ rounds in $\bm{v}^{\prime}$, the transformation preserves the PACE allocation.
\begin{lemma}
    \label{lemma-auxiliary-1}
    If $s^\prime < t^*$, then for agent $j \in \{2, \cdots, n\}$,
    \begin{equation*}
         U_j^{s^\prime}(\bm v^\prime)< \ell s^{\prime}  \implies \frac{\Hat{U}_j^{s^\prime}(\bm{v}^\prime)}{\Hat{U}_1^{s^\prime}(\bm{v}^\prime)}<\varepsilon.   
    \end{equation*}
    Also, all items in the first $t^*$ rounds of $\bm{v}^\prime$ has the same allocation result as $\bm v$.
\end{lemma}

For the second part, we show that in the first $t^*$ rounds, the problem of maximizing envy under PACE can be relaxed into a canonical program defined as \eqref{eq-optimization-canonical}. \Cref{lemma-auxiliary-2} gives upper bounds on multiplicative envy in the first $t^*$ rounds, as well as the number of rounds within $t^*$ for agent $j\in\{2, \cdots, n\}$ to have non-zero values.
\begin{lemma}
\label{lemma-auxiliary-2}
    For any agent $j\in \{2, \cdots, n\}$, let $B_j$ be the set of items that belongs to $A_j$ in the original sequence $\bm v$, and are permuted to the first $t^*$ position in the transformed sequence. Then it holds that
    \begin{equation*}
        \lim_{S\rightarrow \infty} \frac{U_1(B_j)}{U_1(A_1)} \leq 1 + 2\log \frac{1}{\varepsilon}, \ 
        \lim_{S\rightarrow \infty} \frac{|B_j|}{S} \leq \frac{1}{\varepsilon}\left(1+\log \frac{1}{\varepsilon}\right).
    \end{equation*}
\end{lemma}

For the final part, we show that the length of the transformed sequence cannot exceed $t^*$. Let $m:=\frac{S}{\varepsilon}\left(1+\log \frac{1}{\varepsilon}\right)$ be the asymptotic upper bound of $|B_j|$ given in \Cref{lemma-auxiliary-2}. It is straightforward to show with some simple algebra that $S+(n-1)m \leq t^*$. Notice that $S+(n-1)m$ is an asymptotic upper bound on the number of non-empty items (an empty item is one that is evaluated zero by all buyers) within $t^*$ rounds in the transformed sequence. Therefore, $S+(n-1)m \leq t^*$ implies that the transformed sequence has to end within $t^*$ rounds. Thus $t^*$ is an upper bound of both the original and transformed sequence, whose allocation results are guaranteed to be the same by \Cref{lemma-auxiliary-1}. The upper bound for envy in \Cref{lemma-auxiliary-2} is then a valid upper bound for the original sequence. 

$\hfill \square$

\subsection{Nash Welfare Analysis for PACE}
\label{section-nw-pace}
Next we focus on establishing a worst-case guarantee for the asymptotic competitive ratio of Nash Welfare for PACE under Assumption \ref{assumption-pace-2}.
To begin with, we show that Assumption \ref{assumption-pace-2} implies infinite agent utilities with appropriate initialization.
\begin{lemma}
\label{lemma-infinite-utility}
    For any input $\bm{v}\in \mathcal{V}_{\varepsilon, c}^T$ (defined in \cref{assumption-pace-2}), if $$ \ell < \frac{c \varepsilon^2}{1 +(n-1)(1+\log 1/\varepsilon)}, \ r=1,$$ 
    then there exists a constant $d>0$ (which depends on $\ell$), such that for each $i\in[n]$, PACE satisfies $\lim_{T\rightarrow \infty} U_i({A_i}) \geq dT.$
    Furthermore, $\frac{d}{c} = \Omega\left(\frac{1}{n}\right)$ as $n\rightarrow \infty$.
\end{lemma}
The proof of \Cref{lemma-infinite-utility} is also by a reduction from the canonical problem \eqref{eq-optimization-canonical}. We remark that \Cref{lemma-infinite-utility} yields more than infinite agent utilities: it also tells us that the utilities are linear in $T$. Moreover, from $d/c = \Omega\left(n^{-1}\right)$ we know that PACE computes an asymptotic approximate proportional allocation, which helps to derive a bounded competitive ratio \textit{w.r.t.} Nash welfare. The bound can be furthermore refined using the envy results. 

\begin{theorem}[Upper Bound of Competitive Ratio for PACE]
\label{thm-cr-pace}
    Under Assumption \ref{assumption-pace-2}, for any input $\bm{v}\in \mathcal{V}_{\varepsilon, c}^T$, if $$ \ell < \frac{c \varepsilon^2}{1 +\varepsilon+(n-1)(1+\log 1/\varepsilon)}, \ r=1,$$ 
    PACE achieves
$$  \lim_{T\rightarrow \infty} \sup_{v \in \mathcal{V}^T_\varepsilon} \left(\frac{\prod_{i=1}^n U_i^\star}{\prod_{i=1}^n U_i(A_i)} \right)^{1/n} \leq \left(1+2\log \frac{1}{\varepsilon}\right) \cdot \frac{1}{c}.$$
\end{theorem}

\Cref{thm-cr-pace} gives an upper bound depends only on parameters $\varepsilon$ and $c$, which are both independent of $T$. We remark that the constant $1/c$ might not be tight. 

We also remark that in both \Cref{thm-envy-pace} and \Cref{thm-cr-pace} $\ell$ is at most the order of $1/n$, which is aligned with the stationary and non-stationary setting (projecting utilities to an $\Omega(1/n)$ bound is unreasonable since there are $n$ agents). However with adversarial input $\ell$ decreases as $\varepsilon\rightarrow 0$, which means that PACE requires a wider projection interval when its input becomes potentially more extreme.

\section{Conclusion}
\label{section-discussions}
We proved horizon-independent bounds for envy and Nash welfare for the integral greedy algorithm and PACE under adversarial inputs with mild assumptions. Our results complete the first best-of-many-worlds result for online fair allocation, since PACE thus achieves guarantees under stochastic~\citep{gao2021online}, stochastic but nonstationary~\citep{liao2022nonstationary}, and adversarial inputs.
Moreover, our results on integral greedy are of independent interest, as they characterize assumptions needed to achieve guarantees for that algorithm.

It remains open whether the constant in \Cref{thm-cr-pace} can be improved. A more general open question is to explore more best-of-many-worlds online fair allocation algorithms, with potentially different performance measures and assumptions.
%
%
%
\bibliographystyle{plainnat}
\bibliography{main}
\appendix
\newpage
\section{Additional Related Work}
\label{appendix-related-work}
\paragraph{Offline Nash Welfare Approximation}
To explore PACE's behavior under adversarial inputs, we focus on the approximation of the EG objective, which is equivalent to maximizing Nash welfare. Nash welfare was introduced by \citet{nash1950bargaining}, and is one of the ideal proxies for balancing fairness and efficiency in allocation problems ~\citep{kaneko1979nash}. Maximizing Nash welfare is well-studied in the offline settings. For divisible items, it is equivalent to the Eisenberg-Gale convex program ~\citep{eisenberg1959consensus}. For indivisible items, computing a Nash-welfare-maximizing allocation is APX-hard~\citet{lee2017apx} for additive utilities. \citet{cole2018approximating} gave the first approximation algorithm, with $2e^{1/e}$ competitive ratio. \citet{anari2016nash} applied matrix permanent and stable polynomials to get an $e$-approximation. \citet{cole2017convex} improved the ratio to $2$. The state-of-the-art algorithm, given subsequently by \citet{barman2018finding}, has a ratio of $e^{1/e}$. 
More general utility classes have been considered beyond additive ones. \citet{garg2018approximating} achieved a $2e^{1/e}+o(1)$ ratio under budget-additive values. \citet{anari2018nash} considered separable, piecewise-linear and concave valuations and gave an $e^2$-competitive algorithm. For submodular utilities, \citet{garg2020approximating} achieved an $O(n\log n)$-approximation, which was recently improved to a constant ratio~\citep{li2022constant}. For more generalized forms of submodular utilities, an $O(n)$ ratio was reached by \citet{barman2020tight} and \citet{chaudhury2021fair} independently.

\paragraph{Online Allocation with Envy Guarantees}
Besides Nash welfare maximization, our analysis on multiplicative envy is also related to the line of works focused on achieving (possibly approximate) envy-freeness in the online setting. \citet{bogomolnaia2022fair} assume stochastic input and enforce envy-freeness as a soft constraint, while maximizing social welfare. \citet{he2019achieving} further allow reallocating previous items, and show that $O(T)$ reallocation is enough to achieve envy-freeness up to an item. \citet{benade2018make} considers envy minimization in the stochastic, indivisible setting, and show that allocating each item to each agent uniformly at random is near-optimal up to logarithmic factors. \citet{zeng2020fairness} considers the indivisible setting with non-adaptive adversary, showing that nontrivial approximation of envy-freeness and Pareto-optimality is hard to achieve simultaneously. In contrast to the above works, we focus on \textit{multiplicative} envy instead of additive envy. While it is shown by \citet{caragiannis2019unreasonable} that Nash welfare maximizing allocation has approximate envy-freeness, our analysis on multiplicative envy is independent of other fairness measures.
Finally note that the PACE algorithm, while focused on the stronger guarantee of asymptotically maximizing Nash welfare in the stochastic setting, actually achieves asymptotic envy-freeness as well, since it converges to the equilibrium allocation of the underlying Fisher market.
\section{Missing Proofs in \Cref{section-greedy}}
\label{appendix-greedy}
\subsection{Proof of \Cref{thm-envy}}
\label{proof-thm-envy}
According to \Cref{lemma-induction-property}, it suffices to bound the envy of all $2$--agent instances. We will show that for any input $\bm v \in \mathcal{V}_{\varepsilon}^T$, \begin{equation}
    \frac{B_2 U_2(A_1)}{B_1 U_2(A_2)}\leq 1 + 2\log \frac{1}{\varepsilon} + \frac{(1+\varepsilon^4)(1+\varepsilon^2)}{\varepsilon^6} \frac{1}{T} + \frac{(1+\varepsilon^2)^2}{\varepsilon^8} \frac{1}{T^2}
\end{equation}
Consider the following transformation operation from instance $\bm{v}$ into ${\bm{v}}^\prime$:
    \begin{enumerate}
        \item Set agent $1$'s valuation for all items in $A_2$ to $0$. 
        \item Move all items in $A_2$ to the beginning of the input sequence, and all items in $A_1$ to the end of the input sequence (in arbitrary order).
    \end{enumerate}
We claim that the allocation result under the budget-adapted greedy allocation algorithm is invariant to this transformation. 
Clearly, all items in $A_2$ will still be allocated to agent $2$, since we've set agent $1$'s valuation on them to zero. For item $t \in A_1$, let $t^\prime$ be the round it appears in $\bm v^{\prime}$. One can check
\begin{equation*}
     \frac{B_1 v_1^t}{B_2 v_2^t} \geq \frac{U_1^t(\bm v)}{U_2^t(\bm v)} \geq \frac{U_1^{t^\prime}(\bm{v}^\prime)}{U_2^{t^\prime}(\bm{v}^\prime)},
\end{equation*}
where the first inequality is due to $t\in A_1$, and the second inequality is because $U_1^{t^\prime}(\bm{v}^\prime) = U_1^{t}(\bm{v})$ and $U_2^{t^\prime}(\bm{v}^\prime) = \sum_{s=1}^T x_2^s v_2^s \geq U_1^{t}(\bm{v})$.

Therefore, it suffices to consider only the transformed sequences, i.e., ones with agent $2$ getting every item in the first $R$ rounds, and then nothing thereafter. To find the transformed sequence of maximum envy, we next note that we can increase agent $2$'s valuation on the items occurring after round $R$ as much as possible, while ensuring that none of them are allocated to agent $2$ by the algorithm, while weakly increasing envy. This is equivalent to the following optimization problem, with the multiplicative envy of agent $2$ upper bounded by its objective:
\begin{equation*}
\label{eq-optimization-envy}
    \begin{aligned}
    \max_{v_1^t, v_2^t} & \ \ \frac{B_2}{B_1 V_2^0}\sum_{t=R+1}^{\infty} v_2^t \\
    \text{s.t.} & \ \ v_1^t, v_2^t \in \{0\}\cup[\varepsilon, 1] \\
                & \ \ \frac{v_1^t}{v_2^t} \geq \frac{B_2 U_1^{t-1}}{B_1 V_2^0} \\
                & \ \ U_1^t = \sum_{s=R+1}^{t} v_{1}^s \\
    \end{aligned}
\end{equation*}

Notice the constraint that all items after round $R$ are allocated to agent $1$. Together with the $\{0\}\cup[\varepsilon, 1]$ range for values, this gives an upper bound of $v_2^t/v_1^t$ on each round:
\begin{equation*}
    \frac{v_2^t}{v_1^t}\leq q(U_1^{t-1}) := \begin{cases}
        \min\left\{\frac{B_1 V_2^0}{B_2 U_1^{t-1}}, \frac{1}{\varepsilon}\right\} & U_1^{t-1} \leq \frac{B_1}{B_2\varepsilon}V_2^0\\
        0 & U_1^{t-1} > \frac{B_1}{B_2\varepsilon}V_2^0
    \end{cases}
\end{equation*}

The multiplicative envy can then be further upper bounded by
\begin{equation}
\label{eq-upper-bound-q}
    \frac{B_2 U_2(A_1)}{B_1 U_2(A_2)} \leq \frac{B_2}{B_1 V_2^0} \sum_{t=R+1}^{R+S} v_1^t q(U_1^{t-1})
\end{equation}
where $S$ is the number of rounds in the above optimization problem with $v_1^t, v_2^t > 0$. Since we should not allocate anything to agent $2$ after round $R$, we have $U_1^{S} \in \left(0,\frac{B_1 V_2^0}{B_2 \varepsilon}+1\right]$.

Next, we show that the right hand side of (\ref{eq-upper-bound-q}) can be approximated with an integration. We have
\begin{align}
     &\frac{B_2}{B_1 V_2^0}\left(\sum_{t=R+1}^{R+S} v_1^t q(U_1^{t-1}) - \int_{0}^{U_1^{R+S}} q(u)\mathrm{d} u\right) \nonumber\\
     =&\frac{B_2}{B_1 V_2^0} \sum_{t=R+1}^{R+S} \left(v_1^t q(U_1^{t-1}) - \int_{U_1^{t-1}}^{U_1^t} q(u)\mathrm{d}u\right) \nonumber\\
     =&\frac{B_2}{B_1 V_2^0} \sum_{t=R+1}^{R+S} \int_{U_1^{t-1}}^{U_1^t} \left(q(U_1^{t-1} ) - q(u)\right) \mathrm{d} u \nonumber\\
     \overset{\text{(a)}}{\leq}& \frac{B_2}{B_1 V_2^0}\sum_{t=R+1}^{R+S} (v_1^t)\left(q(U_1^{t-1})- q(U_1^t)\right)\nonumber\\
     \overset{\text{(b)}}{\leq}& \frac{B_2}{B_1 V_2^0}\sum_{t=R+1}^{R+S} \frac{B_1}{B_2 \varepsilon^2 V_2^0} (v_1^t)^2 \nonumber\\
     =& \frac{1}{\varepsilon^2 (V_2^0)^2} \sum_{t=R+1}^{R+S} (v_1^t)^2\nonumber\\
     \overset{\text{(c)}}{\leq}& \frac{U_1^{R+S}}{\varepsilon^2 (V_2^0)^2} \nonumber
\end{align}
where (a) is because $q(u)$ is non--increasing, (b) is because the right--side derivative of $q(u)$ is upper bounded by $\frac{B_1 }{B_2 \varepsilon^2 V_2^0}$, and (c) is because $v_1^t \leq 1$. Combined with \ref{eq-upper-bound-q} we have
\begin{equation}
    \frac{B_2 U_2(A_1)}{B_1 U_2(A_2)} \leq \frac{B_2}{B_1 V_2^0}\int_{0}^{U_1^{R+S}} q(u)\mathrm{d} u + \frac{U_1^{R+S}}{\varepsilon^2 (V_2^0)^2} \\
\end{equation}
We then deal with $U_1^{R+S}$ using the relation $T\geq R+S$. This is true because we allow no ``null'' items in a $T$--round instance. From $\varepsilon S \leq U_1^{R+S} \leq V_2^0/\varepsilon +1 \leq R/\varepsilon +1$ we know $V_2^0\geq \frac{\varepsilon^3}{ 1+\varepsilon^2}T$. Hence, 
\begin{align*}
    \frac{B_2 U_2(A_1)}{B_1 U_2(A_2)} &\leq \frac{B_2}{B_1 V_2^0}\int_{0}^{V_2^0/\varepsilon +1} q(u)\mathrm{d} u + \frac{1}{\varepsilon^2 V_2^0} \left(\frac{1}{\varepsilon}+\frac{1}{V_2^0}\right)\\
    &=  \frac{B_2}{B_1 V_2^0}\int_{0}^{V_2^0/\varepsilon} q(u)\mathrm{d} u + \frac{B_2}{B_1 V_2^0}\int_{V_2^0/\varepsilon}^{V_2^0/\varepsilon + 1} q(u)\mathrm{d} u +\frac{1}{\varepsilon^2 V_2^0} \left(\frac{1}{\varepsilon}+\frac{1}{V_2^0}\right)\\
    &\leq 1 + 2\log \frac{1}{\varepsilon} + \frac{\varepsilon}{V_2^0} + \frac{1}{\varepsilon^2 V_2^0} \left(\frac{1}{\varepsilon}+\frac{1}{V_2^0}\right)\\
    &\leq 1 + 2\log \frac{1}{\varepsilon} + \frac{(1+\varepsilon^4)(1+\varepsilon^2)}{\varepsilon^6} \frac{1}{T} + \frac{(1+\varepsilon^2)^2}{\varepsilon^8} \frac{1}{T^2}\\
    &= 1 + 2\log \frac{1}{\varepsilon}+ O\left(\frac{1}{T}\right)
\end{align*}
This finishes our proof. 
\subsection{Proof of \Cref{lemma-number-of-items}}
   For the same reason as \Cref{thm-envy}, it suffices to consider $j=1, i=2$ in $2$-agent instances, with all inputs transformed as described in the proof of \Cref{thm-envy}. $U$ is the short for $U_2(A_2)$, and also the parameter for the canonical optimization problem.

    Consider the two phases in the horizon each with length $T_1$ and $T_2$. We upper bound them respectively:
    \begin{enumerate}
        \item Rounds with $U_1^t \leq U$. Since the increment of $v_1^t$ is at least $\varepsilon$ whenever agent $1$ receives an item, there are at most $U/\varepsilon$ rounds in this phase.
        \begin{equation*}
            \lim_{U\rightarrow \infty} \frac{T_1}{U} \leq \frac{1}{\varepsilon}.
        \end{equation*}
        \item Rounds with $U<U_1^t \leq U/\varepsilon$. Assume this phase begins at round $s$. By the analysis of the canonical optimization problem \eqref{eq-optimization-canonical}, agent $2$'s envy in this phase can be upper bounded asymptotically. 
        \begin{equation*}
            \frac{1}{U}\sum_{t=s}^{s+T_2-1} v_2^t \leq \frac{1}{U}\sum_{t=s}^{\infty} v_1^t q(U_1^{t-1}) \rightarrow \frac{1}{U} \int_{U}^{U/\varepsilon} q(U_1)\mathrm{d} U_1 = 1+\log \frac{1}{\varepsilon}.
        \end{equation*}
        Since $v_2^t$ is at least $\varepsilon$ in each round, 
        \begin{equation*}
            \lim_{U\rightarrow \infty} \frac{T_2}{U} \leq \frac{1}{\varepsilon} \leq \lim_{U\rightarrow \infty} \frac{\sum_{t=s}^{s+T_2-1}v_2^t}{U} \cdot \frac{1}{\varepsilon} \leq \frac{1}{\varepsilon}\left(1+\log \frac{1}{\varepsilon}\right).
        \end{equation*}
    \end{enumerate}
    Combining the bounds for $T_1$ and $T_2$ proves \Cref{lemma-number-of-items}.

\subsection{Proof of \Cref{thm-envy-lower}}

    We prove \Cref{thm-envy-lower} by constructing a hard $2$--agent instance that reaches $1+2\log \frac{1}{\varepsilon}$ envy asymptotically. The construction follows the spirit of the hard instances described in the proof of \Cref{thm-envy}, where agent $2$ receives items in the beginning $T_0$ rounds, but nothing afterwards.

 We construct a hard instance for $B_1 = B_2$. For unequal weights the construction is similar. For $T_0, a>0$, consider an instance with $2k+2$ phases. In each phase, the same item appears for many rounds. The agent's valuations on these items are described in Table \ref{table-envy}.   
\begin{table}[t]
\label{table-envy}
\centering
\begin{tabular}{  | l | l | l | l | l | l | l }
\hline
  Phase & Length & Valuation 1 & Valuation 2 & $U_1$ (after the phase) & $U_2$ (after the phase)\\ \hline
  & & & & &  \\
A1& $T_0$ & $0$ & $1$ & $0$ & $T_0$ \\ 
& &  &  & &\\\hline
  & & & & &  \\
A2& $T_0$ & $\varepsilon$ & $1$ & $\varepsilon \cdot T_0$ & $T_0$ \\ 
& &  &  & &\\\hline
  & & & & &  \\
B1& $\left(1 - \frac{1}{a}\right) T_0$ & $\varepsilon \cdot a^1$ & $1$ & $\varepsilon\cdot a^1 \cdot T_0$ & $T_0$ \\ 
& &  &  & &\\\hline
  & & & & &  \\
B2& $\left(1 - \frac{1}{a}\right) T_0$ & $\varepsilon \cdot a^2$ & $1$ & $\varepsilon\cdot a^2 \cdot T_0$ & $T_0$ \\ 
& &  &  & &\\\hline
...& ... & ... & ... & ... & ... \\ \hline
  & & & & &  \\
B$k$& $\left(1 - \frac{1}{a}\right) T_0$ & $\varepsilon \cdot a^k = 1$ & $1$ & $\varepsilon\cdot a^k \cdot T_0 = T_0$  & $T_0$ \\ 
& &  &  & &\\\hline
  & & & & &  \\
C1& $ \frac{1}{\varepsilon}\left(1 - \frac{1}{a}\right) T_0$ & $\varepsilon \cdot a^1$ & $\varepsilon$ & $a^1 \cdot T_0$ & $T_0$ \\ 
& &  &  & &\\\hline
  & & & & &  \\
C2& $ \frac{1}{\varepsilon}\left(1 - \frac{1}{a}\right) T_0$ & $\varepsilon \cdot a^2$ & $\varepsilon$ & $a^2 \cdot T_0$ & $T_0$ \\ 
& &  &  & &\\\hline
...& ... & ... & ... & ... & ... \\ \hline
  & & & & &  \\
C$k$& $ \frac{1}{\varepsilon}\left(1 - \frac{1}{a}\right) T_0$ & $\varepsilon \cdot a^k = 1$ & $\varepsilon$ & $a^k \cdot T_0 = \frac{1}{\varepsilon} T_0$ & $T_0$ \\ 
& &  &  & &\\ \hline
\end{tabular}
\caption{A worst-case instance for envy}
\end{table}

One can check that the boundary condition $v_1^t/v_2^t = U_1^t /U_2^t$ holds at the end of each phase, and all items after phase A1 is allocated to agent $1$. Agent $2$'s monopolistic utility is 
$$ V_2 = 2\left( 1+ \left(1-\frac{1}{a}\right)\cdot \log_{a}{\frac{1}{\varepsilon}}\right) T_0$$
As $T_0 \rightarrow \infty$, we consider $a \rightarrow 1$, we have
$$\left(1-\frac{1}{a}\right)\cdot \log_a{\frac{1}{\varepsilon}} \rightarrow \log \frac{1}{\varepsilon}, \ \frac{V_2}{T_0} = 2\left(1+\log \frac{1}{\varepsilon}\right).$$
Out of his monopolistic utility, agent $2$ only gets utility $T_0$. This gives $1+2\log 1/\varepsilon$ multiplicative envy asymptotically.

\subsection{Proof of \Cref{thm-cr-upper}, $\varepsilon=1$}
Without loss of generality, we assume $U_1(A_1) \leq U_2(A_2) \leq \cdots \leq U_n(A_n)$. Our main proof strategy is to show that for any input instance $\bm v \in \mathcal{V}_{1}^T$, we have for $i\in [n]$, 
\begin{equation}
\label{eq-n-i+1}
    \lim_{T\rightarrow \infty} \frac{U_i^\star}{U_i(A_i)}\leq n-i+1.
\end{equation}
For $i\in [n]$, define $S_i\subseteq [T]$ as the set of items on which at least one agent in $\{1, 2, \cdots, i\}$ has nonzero valuation,
\[S_i = \{t: \ \exists \ j \in [i], \ v_j^t=  1\}.\]
Consider agent $k$ where $k>i$. Notice that the integral greedy algorithm guarantees that the number of items in $S_i$ which are allocated to agent $k$ is strictly upper bounded by $U_i(A_i)+1$. If not, then the last such item $t$ should have been allocated to some agent $j\leq i$ since $U_{k}^{t-1}> U_i(A_i)>U_j(A_j)\geq U_j^{t-1}$, which is a contradiction. Therefore, we get an upper bound on the amount of items that buyer $k>i$ receives in $S_i$:
\begin{equation*}
    |A_{k} \cap S_i| \leq U_i(A_i) \ \ (1\leq i \leq {n-1}, k> i).
\end{equation*}
Define the item set $D_i = \bigcup_{k\geq i}(A_k\cap S_i)$. Then, 
\begin{equation}
\label{eq-k-on-Si}
    |D_i| \leq (n-i+1)U_i(A_i) + (n-i).
\end{equation}
Intuitively, it is possible that in an offline optimal allocation, all items in $D_i$ should be completely shared among the first $i$ agents, i.e., nothing in $D_i$ is allocated to agent $k>i$. To verify this, simply consider a case where $|S_i|$ is negligibly small compared with $|S_{i+1}|$. Further, it is also possible that the all items in $D$ are given to agent $i$ in an optimal allocation, for example, in the case where buyer $1, 2, \cdots, i-1$ have zero valuation on these items. This will lead to ${U_i^\star} = U_i(A_i)+(n-i)(U_i(A_i)+1)$. 

However, we will show that $U_i^\star$ can not be any larger by contradiction. Suppose $U_i(A_i^{\star}) \geq (n-i+1)(U_i(A_i)+1)$. By inequality \eqref{eq-k-on-Si}, at least one item in $\Bar{D}_i = \bigcup_{j<i }(A_j\cap S_{i-1})$ is partly allocated to agent $i$ in $\bm{x}_i^\star$. 
This means that an item that belongs to agent $j$ $(j<i)$ in the integral greedy allocation should be re-allocated to agent $i$ in the offline optimal allocation, even with a large $U_i^\star> U_j^\star$. This is impossible for an optimal allocation, since allocating this item to agent $j$ strictly increases the geometric mean. 

Therefore,
$$U_i(A_i^{\star})/U_i(A_i) \leq (n-i+1) +\frac{1}{U_i(A_i^{\star})}.$$
Combining this with the fact that the optimal offline allocation has $U_i(A_i) \rightarrow \infty \ (T\rightarrow \infty)$ for all $i\in [n]$, this proves (\ref{eq-n-i+1}).
    $\hfill \square$

\subsection{Proof of \Cref{thm-cr-upper}, $\varepsilon \in (0,1)$}
\label{proof-thm-cr-upper}

Without loss of generality, we assume the integral greedy algorithm gives us $U_1(A_1) \leq U_2(A_2) \leq \cdots \leq U_n(A_n)$. To prove \Cref{thm-cr-upper}, it suffices to show that for any $\alpha>0$ there exists a constant $\lambda$ such that for all $i\in [n]$ and any input instance in $\mathcal{V}_{\varepsilon}^T$,
\begin{equation}
\label{eq-upper-n-i}
    \lim_{T\rightarrow \infty} \frac{U_i^\star}{U_i(A_i)}\leq \lambda \cdot i^\alpha \cdot\left(1+(n-i)\left(1+2\log \frac{1}{\varepsilon}\right)\right).
\end{equation}

We introduce two lemmas to prove \eqref{eq-upper-n-i}. 
\Cref{lemma-envy-generalized} shows that for a set of item $A \subset A_i$ which belongs to agent $i$ in the greedy algorithm, there is an upper bound on the increment of welfare if one transfers $A$ to other agents. It is a generalization of the envy result. The proof of \ref{lemma-envy-generalized} is similar to \ref{thm-envy}, and is deferred to \Cref{proof-lemma-envy-generalized}.
\begin{lemma}
\label{lemma-envy-generalized}
For $i\in [n]$, let $A \subseteq A_i$ be a set of items that are allocated to agent $i$ by the integral greedy algorithm, such that $|A|\rightarrow \infty$ as $T\rightarrow \infty$. Let $J \subseteq [n]\backslash\{i\}$ be a set of agents not including $i$, and $U = \max_{j \in J} U_j(A_j)$. Then for any $\bm{v} \in \mathcal{V}_{\varepsilon}^T$, 
    $$\lim_{T\rightarrow \infty} 
    \frac{\sum_{s \in A} \max_{j \in J} v_j^s}{U} \leq \int_{0}^{\min\left\{\frac{1}{\varepsilon}, \frac{U_1(A)}{U}\right\}} r(u) \mathrm{d}u ,$$
where \begin{equation*}
    r(u) = \begin{cases}
        \frac{1}{\varepsilon} & 0\leq u \leq \varepsilon\\
        \frac{1}{u} & \varepsilon<u\leq \frac{1}{\varepsilon} \\
        0 & u>\frac{1}{\varepsilon}
    \end{cases}.
\end{equation*}
\end{lemma}

The next lemma guarantees that asymptotically, the utility of any single agent in the offline optimal allocation can be at most $\lambda$ times as large as the utility of the ``wealthiest'' agent in the integral greedy allocation. $\lambda$ is a positive constant which does not depend on $T$ and $n$. The proof of \Cref{lemma-balanced} can be found in \Cref{proof-lemma-envy-balanced}.
\begin{lemma}
    \label{lemma-balanced}
For any given input $\bm{v} \in \mathcal{V}_{\varepsilon}^T$ and $\alpha>0$, there exists a constant $\lambda >0$ such that for any $k \in [n]$,

$$ \lim_{T\rightarrow \infty} \frac{U_k^\star}{\max_{j\in [n]} U_j(A_j)}\leq \lambda n^\alpha.$$
\end{lemma}

After introducing the two auxiliary lemmas, we return to the proof of \Cref{thm-cr-upper}.

For $i\in [n]$, define $S_i\subseteq [T]$ as the set of items on which at least one agent in $\{1, 2, \cdots, i\}$ has nonzero valuation. 
\[S_i = \{t: \ \exists \ j \in [i], \ v_j^t=  1\}.\]

For agent $k\geq i$, we consider $A_k \cap S_i$, which is the set of items in $S_i$ that are allocated to agent $k$. 

In the proof sketch, we give the intuition that, it is possible that the optimal allocation gives all items in $A_k \cap S_i, (k>i)$ to agent the first $i$ agents, which might cause a large ratio of $U_i^\star / U_i(A_i)$. 
To bound the ratio, we first apply \Cref{lemma-envy-generalized} to bound the maximum welfare of the first $i$ agents on item set $A_k \cap S_i$. For $k>i$,
\begin{equation}
\label{eq-k-on-Si-1}
    \lim_{T\rightarrow \infty} 
    \frac{\sum_{s \in S_i\cap A_k} \max_{j \leq i} v_j^s}{U_i} \leq 1+2\log \frac{1}{\varepsilon}.
\end{equation}
Then, consider the input $\bm{v}^\prime$, which is defined as setting agent $k$'s value to $0$ on every item in $S_i$. Because agent $1, 2, \cdots, i$ are given strictly more items in $\bm{v}^\prime$, both the integral greedy algorithms should have non-decreasing agent utilities when $T$ is large enough. (The amount of decrease should be negligible in $T$).\begin{equation}
\label{eq-monotone}
\begin{aligned}
   \lim_{T\rightarrow \infty} \frac{U_j^\star(\bm v)}{U_i(\bm{v})} \leq \lim_{T\rightarrow \infty} \frac{U_j^\star(\bm{v}^\prime)}{U_i(\bm{v})}, \forall j\in [i]\\
    \lim_{T\rightarrow \infty} \frac{U_j(\bm v)}{U_i(\bm{v})} \leq \lim_{T\rightarrow \infty} \frac{U_j(\bm{v}^\prime)}{U_i(\bm{v})}, \forall j\in [i] 
\end{aligned}
\end{equation}
Consider the integral greedy allocation after transforming from $\bm{v}$ to $\bm{v}^\prime$. By equation \eqref{eq-k-on-Si-1}, the total utility increase of the $i$ agents is asymptotically upper bounded by $(n-i)\left(1+2\log \frac{1}{\varepsilon}\right)$. Then the utility of the ``wealthiest'' agent after the transformation should be upper bounded by
\begin{equation}
    \label{eq-monotone-application}
    \lim_{T\rightarrow \infty} \frac{\max_{j\in [i]} U_j(\bm{v}^\prime)}{U_i(\bm{v})} \leq 1+ (n-i)\left(1+2\log \frac{1}{\varepsilon}\right).
\end{equation}
If \eqref{eq-monotone-application} does not hold, this will contradict the the non-decreasing property in \eqref{eq-monotone}.

Combining \eqref{eq-monotone-application} and the first inequality in \eqref{eq-monotone}, 

\begin{align*}
    \lim_{T\rightarrow \infty} \frac{U_i^\star(\bm v)}{U_i(\bm{v})} &\leq \lim_{T\rightarrow \infty} \frac{U_i^\star(\bm{v}^\prime)}{U_i(\bm{v})}\\
    &= \lim_{T\rightarrow \infty} \frac{\max_{j\in [i]}U_j(\bm{v}^\prime)}{U_i(\bm{v})} \cdot \lim_{T\rightarrow \infty} \frac{U_i^\star(\bm{v}^\prime)}{\max_{j\in [i]}U_j(\bm{v}^\prime)} \\
    &\leq \left(1+ (n-i)\left(1+2\log \frac{1}{\varepsilon}\right)\right)\lim_{T\rightarrow \infty} \frac{U_i^\star(\bm{v}^\prime)}{\max_{j\in [i]}U_j(\bm{v}^\prime)} \\
    &\overset{\text{(a)}}{\leq} \lambda \cdot i^{\alpha} \cdot \left(1+ (n-i)\left(1+2\log \frac{1}{\varepsilon}\right)\right).
\end{align*}
In the last step (a), we applied \Cref{lemma-balanced} to $\bm{v}^\prime$, since $\bm{v}^\prime$ can be regarded as an input sequence with $i$ agents. This proves \eqref{eq-upper-n-i}.

\subsection{Proof of \Cref{lemma-envy-generalized}}
\label{proof-lemma-envy-generalized}
The proof of \ref{lemma-envy-generalized} follows the same pattern as the proof of \Cref{thm-envy}. First, transform the input while preserving the result of integral greedy algorithm. Then, use an optimization program to characterize that gives maximum value to the left-hand side of \ref{lemma-envy-generalized}.

Thanks to \Cref{lemma-induction-property}, we do not need to consider agents that are not in $J\cup \{i\}$, thus assume without loss of generality that $i = n$, $J = \{1, 2, \cdots, n-1\}$ and $U_1(A_1)\leq \cdots \leq U_{n-1}(A_{n-1}) = U$.

For any given input $\bm{v}$, we transform it into $\bm{v}^\prime$ without changing the result of integral greedy allocation:
\begin{enumerate}
    \item For each $j\in\{1, 2, \cdots, n-1\}$ and each item in $A_j$, set all agents' value to $0$ except agent $j$.
    \item Put all items in $A_n$ to the end of the sequence.
\end{enumerate}

Therefore, it suffices to only consider sequences that only allocate items to agent $n$ after other agents have received all items in the beginning $R$ rounds.
\begin{equation*}
\label{eq-optimization-envy-generalized}
    \begin{aligned}
    \max_{v_i^t} & \ \ \frac{1}{U}\sum_{t=R+1}^{\infty} \max_{j<n}v_j^t \\
    \text{s.t.} & \ \ v_i^t \in \{0\}\cup[\varepsilon, 1], \ \ \forall i\in[n] \\
                & \ \ \frac{v_n^t}{v_j^t} \geq \frac{U_n^{t-1}}{U_j(A_j)}, \  \ \forall j \in [n-1]\\
                & \ \ U_n^t = \sum_{s=R+1}^{t} v_{n}^s \\
    \end{aligned}
\end{equation*}
An upper bound of $(\max_{j<n} v_j^t)/v_n^t$ is given by $\min\left\{\frac{1}{\varepsilon}, \frac{U}{U_n^{t-1}}\right\}$. The rest of the analysis is identical to the analysis of the canonical optimization program \eqref{eq-optimization-canonical} in \Cref{thm-envy}.

\subsection{Proof of \Cref{lemma-balanced}}
\label{proof-lemma-envy-balanced}

For simplicity of notation, in this proof we denote $U_j = U_j(A_j)$. 

Consider modifying the integral greedy allocation into an optimal one: for each $i, j \in [n]$, some items that belong to $i$ in the integral greedy allocation might be partly re-allocated to $j$ in the optimal allocation, and vice versa. To characterize this procedure of modification, we define the following variables:
\begin{itemize}
    \item $Z_i := \sum_{t\in A_i}(1-x_i^{\star,t})\cdot v_i^t$ is agent $i$'s value on the items that belong to him in the integral greedy allocation, but not in the optimal allocation.
    \item $Y_{ji}:= \sum_{t\in A_j}x_i^{\star, t} v_i^t$ as agent $i$'s value on the items that belong to agent $j$ in the integral greedy allocation, but now belong to agent $i$ in the optimal allocation. Specifically, we let $Y_{jj} = 0$.
\end{itemize}
From the above definition,
\begin{equation}
\label{eq-in-out}
    U_i^\star = U_i - Z_i  + \sum_{j\neq i} Y_{ji}, \ i \in [n].
\end{equation}

Divide all variables in \eqref{eq-in-out} by $\max_{j\in [n]}U_j(A_j)$ to get $u_i^{\star}, u_i, z_i$ and $y_{ji}$. Then \eqref{eq-in-out} becomes
\begin{equation}
\label{eq-in-out-asymptotic}
    u_i^\star = u_i - z_i  + \sum_{j\neq i} y_{ji}, \ i \in [n].
\end{equation}
All above variables should be non-negative. Applying \Cref{lemma-envy-generalized}, we get further constraints on $y_{ji}$:
\begin{align}
    y_{ji} &\leq u_i \int_{0}^{\min\{\frac{1}{\varepsilon}, \frac{z_j}{u_i}\}}r(u) \mathrm{d} u +o_T(1) \leq c_1 \min\{u_i, u_j\}+o_T(1) \label{condition-1}\\
    \sum_{i \neq j} y_{ji}&\leq \max_{i\neq j} u_i \int_{0}^{\min\{\frac{1}{\varepsilon}, \frac{z_j}{u_i}\}}r(u) \mathrm{d} u +o_T(1) \leq c_2 \min\{\max_{i\neq j} u_i, u_j\} +o_T(1) \label{condition-2}
\end{align}
In the above conditions we use $o_T$ to explicitly indicate that the asymptotic notation is on $T$. $c_1, c_2 \leq 1/\varepsilon$ are constants.

Notice that so far we have not yet used the condition that $(u_1^\star, \cdots, u_n^\star)$ is an optimal allocation. We use a necessary condition (not sufficient) to characterize optimality:
\begin{equation}
\label{eq-optimality-necessasry}
    u_j^\star < \varepsilon u_i^\star \implies y_{ji} = 0 
\end{equation}
Condition \eqref{eq-optimality-necessasry} holds because when $u_j^{\star} < \varepsilon u_i^\star$ in an optimal allocation, agent $j$ must have no envy towards agent $i$. If not, then allocating any envied item to agent $j$ will strictly increase the geometric mean, since all non-negative values are in $[\varepsilon,1]$.

Next, we show that $u_k^\star = o(n^\alpha)$, for any $\alpha>0$. 
If $u_k^{\star}$ has a constant upper bound, this is obviously true. Next, we show that if $u_k^{\star} \geq g(n)$ as $T$ goes to infinity, where $g(n)\rightarrow \infty$ as $n\rightarrow \infty$, we can derive $g(n) = o(n^\alpha)$ for any positive $\alpha$. Notice that the asymptotic notion is on $n$ when we refer to $g(n)$. 

Consider a directed graph $D$ with $[n]$ as its vertex set, and the set of edges is defined as:
\begin{equation*}
    E(D) = \left\{ (i, j): y_{ji} > 0\right\}
\end{equation*}

Define $J(\ell)\subseteq [n]$ as the set of all vertices $j$ such that the distance from $i$ to $j$ is $\ell$, i.e., the shortest directed path from $i$ to $j$ has length $\ell$. $J(0) = \{i\}$. We denote $m_\ell:=|J(\ell)|$.


For $j \in J(\ell)$, there exists a path $j_0 j_1 \cdots j_\ell$ where $j_0 = k, j_\ell = j$. By condition \eqref{eq-optimality-necessasry}, the existence of edge $(j_s, j_r)$ implies $u_{j_r}^\star \geq \varepsilon u_{j_s}^\star$. Thus $u_{j}^\star \geq \varepsilon u_{j-1}^\star \geq \cdots \geq \varepsilon^{\ell} u_k^{\star} = g(n)$. This gives a lower bound on the sum of $u_j^{\star}$ where $j\in J(\ell)$:
\begin{equation}
\label{eq-layer-l-lower}
    \sum_{j \in J(\ell)} u_j^{\star} \geq \varepsilon^\ell \cdot g(n)\cdot m_\ell.
\end{equation}

Next, we show inductively that $m_\ell = \Omega(g(n)\cdot m_{\ell-1}).$ For $\ell = 1$, this is true because as $T\rightarrow \infty$, by condition \eqref{condition-1},
$$ g(n) \leq u_k^\star \leq u_k + \sum_{j \in J(1)}y_{jk} \leq u_k +c_1 m_1.$$

For $\ell \geq 1$, suppose $m_{\ell^\prime} = \Omega(g(n)\cdot m_{\ell^\prime-1}) = \Omega(g^{\ell^\prime}(n))$ is true for $\ell^\prime = 1, 2, \cdots, \ell$. Consider again the sum of all $u_j^{\star}$ where $j\in J(\ell)$, as $T\rightarrow \infty$,
\begin{align}
    \sum_{j\in J(\ell)} u_{j}^\star &\leq \sum_{j\in J(\ell)}u_j + \sum_{j\in J(\ell)}\sum_{i \neq j} y_{ij}  \nonumber \\ \nonumber
    &\overset{\text{(a)}}{=} \sum_{j\in J(\ell)}u_j + \sum_{j\in J(\ell)}\sum_{\ell^\prime = 0}^{\ell+1} \sum_{i \in J(\ell^\prime)}y_{ij}\\ \nonumber
    &=\sum_{j\in J(\ell)}u_j + \sum_{\ell^\prime = 0}^{\ell+1} \sum_{i \in J(\ell^\prime)}\sum_{j\in J(\ell)}y_{ij}\\
    &\overset{\text{(b)}}{\leq} \sum_{j\in J(\ell)}u_j + \sum_{\ell^\prime = 0}^{\ell+1} \sum_{i \in J(\ell^\prime)} c_2 u_i \nonumber \\
    &\leq (1+c_2)m_{\ell} + c_2 m_{\ell+1} + \sum_{\ell^\prime=0}^{\ell+1} c_2{m_\ell^\prime} \nonumber \\
    &\overset{\text{(c)}}{=} O(m_\ell) + c_2 m_{\ell+1}, \label{eq-m-l-induction}
\end{align}
where (a) is because $y_{ij} = 0$ for $i \in J(\ell^\prime)$ if $\ell^\prime \geq \ell +2$, (b) is a direct application of condition \ref{condition-2}, and (c) is by the induction hypothesis. 

Combining \eqref{eq-layer-l-lower} and \eqref{eq-m-l-induction}, $c_2 m_{\ell+1} + O(m_{\ell}) \geq \varepsilon^\ell g(n) m_\ell = \Omega(m_\ell g(n))$. This tells us $m_{\ell+1} = \Omega(g(n) \cdot m_{\ell})$ and thus completes the induction.

Define $\ell^* := \max\{\ell: J(\ell) \text{ is non-empty}\}$, i.e., $\ell^*$ is the length of the longest path starting from $S$. Because $m_{\ell^*} = \Omega\left(g^{\ell^{*}}(n)\right) \leq n$, we have 
\begin{equation}
\label{eq-asymptotic-calculation-1}
    \ell^* = O\left(\frac{\log n}{\log g(n)}\right).
\end{equation}
Meanwhile, by \eqref{eq-m-l-induction}, $m_{\ell^*} = \Omega(\varepsilon^{\ell^*} g(n) m_{\ell^{*}})$. We then know that,
\begin{equation}
\label{eq-asymptotic-calculation-2}
    g(n) = O\left(\frac{1}{\varepsilon^{\ell^*}}\right).
\end{equation}
Combining \eqref{eq-asymptotic-calculation-1} and \eqref{eq-asymptotic-calculation-2} one can derive $g(n) = O(e^{\sqrt{\log n}}) = o(n^{\alpha})$ for any $\alpha >0$.

\subsection{Proof of \Cref{thm-opt}}
\label{proof-thm-opt}
Consider dividing the horizon of length $T$ into $n$ different phases, with length $T_1, \cdots, T_n$, satisfying
$$
\sum_{j=1}^n T_j = T, \ \lim_{T\rightarrow \infty} \frac{T_{j-1}}{T_j} = 0  \ \  (\forall j = \{1,2,\cdots, n\})
$$
For the first phase (i.e. for $t \in \{1, 2, \cdots, T_1\}$), $v_i^t = 1$ for all $i\in [n]$. Because there are $T_1$ items and $n$ buyers, for any online integral algorithm there exists a buyer $i_1$ whose utility after phase 1 satisfies ${U}_{i_1}^{T_1} \leq \frac{T_1}{n}$. 

Recursively, for each phase $k\geq 2$, we set for all $t = T_{k-1}+1, \cdots, T_k$:
\begin{equation*}
    v_i^t =\begin{cases}
        0 & i \in \{i_1, \cdots, i_{k-1}\}\\
        
        1 & \text{o.w.}
    \end{cases},
\end{equation*}
and let $i_k \not \in \{i_1, \cdots, i_{k-1}\}$ be a buyer whose utility after round $T_k$ satisfies ${U}_{i_k}^{T_k} \leq \frac{1}{n}\sum_{j=1}^{k} T_j$.

Notice that the optimal offline allocation will allocate everything in phase $k$ to buyer $i_k$. We have $\lim_{T\rightarrow \infty} U_i^{\star} = T_{n-i+1}$ and $$\lim_{T\rightarrow \infty} \frac{({U}_{i})^{\star, T}}{{U}_{i}^T} \geq \frac{1}{i}.$$ This leads to a ratio at least $ (n!)^{1/n}$.

\subsection{Proof of \Cref{thm:seeded int greedy}}
Define $\vmax=\max_{i,\tau} v^\tau_i $, $\vimax = \max_\tau \vtaui$, and $\bar v_{\mathrm{max}} = \frac1n\sumiton v_{i,\mathrm{max}}$.
In this proof, we study the general case with $\vitau \in [0, v_{\mathrm{max}}]$. 
We will prove that
\begin{align}\label{eq:thm seeded}
       R_\delta(v) \leq 3 + \frac{4}{\delta} \bar v_{\mathrm{max}} + 2\log\Big(1+\frac{\vmax}{\delta}\Big) + 2 \log T \;.
\end{align}

The proof is divided into two major steps.
\subsubsection{Step 1. Ratio bounded by prices.} 
We first introduce notations for quantities that appear during the run of \Cref{alg-greedy-seeded}.

Let $\ptau = \max_i \beta ^\tau_i \vtaui = {v^\tau_\itau}/ ({U^{\tau-1}_\itau + \delta})$ be the \textit{price} of item $\tau$, and 
$\beta^\tau _i = \frac{ 1 }{\delta + \Utaumi}$.  Then \Cref{alg-greedy-seeded}
\Cref{eq:seeded_int_greedy} can be written as $\itau = \min \arg \max_i \{ \beta_{i}^{\tau} v_{i}^\tau \}$. Also let $\utaui = \vitau 1( i = i^\tau)$ be the utility of agent $i$ from item $\tau$ according to the allocation of the seeded integral greedy algorithm.
Let $U^\tau_i = \sum_{s=1,\dots,\tau} u^s_i$ be the cumulative utilities for $\tau \geq 1$ and $U^0_i = 0$ and $U_i = U_i^T$.

We begin with the following claim.

\begin{claim}
For any feasible hindsight allocation $x^*$ and its resulting utilities $\widetilde{U}_i$, 
\begin{align}
    \frac1n \sumiton \frac{\widetilde{U}_i + \delta}{U_i + \delta} \leq 1+  \frac1n  \sumtau \ptau  = 1+\frac1n \sumiton \sumtau \frac{\utaui}{ \Utaumi + \delta} \;.
    \label{eq:seeded bound by priced}
\end{align}
\end{claim} 

 Notice that $\frac1n \sumiton \frac{\widetilde{U}_i + \delta}{U_i + \delta} \leq \frac1n \sumiton \frac{\widetilde{U}_i }{U_i + \delta} + 1$. Then, 
  \$
  & \frac1n \sumiton \frac{\widetilde{U}_i }{U_i + \delta}
  \\
  & = \frac1n \sumiton \bigg(\frac{ \sumtau \xsttaui \vtaui }{ U_i + \delta} \bigg)
  \\
  & = \frac1n \sumiton \sumtau \bigg(\frac{ \vtaui }{ U_i + \delta} \cdot  \xsttaui \bigg)
  \\
  & \leq \frac1n \sumiton \sumtau \bigg(\frac{ \vtaui }{ {\Utaumi }+ \delta} \cdot  \xsttaui \bigg) =: (\text{A}) \;. 
  \$
By the definition of $\ptau$, it holds for all agent $i$,
\$
\frac{ \vtaui }{ {\Utaumi }+ \delta}  \leq \ptau   \; .
\$
Applying the supply feasibility constraint,
  \$
  (\mathrm{A})  \leq \frac1n \sumiton \sumtau \ptau \cdot  \xsttaui 
 \leq  \frac1n \sumtau \ptau \; .
 \$
Next, use integrality of \Cref{alg-greedy-seeded} and rewrite the sum of prices as the sum of running ratios:
 \$
 & \frac1n \sumtau \ptau
    \\
  & = \frac1n \sumtau \frac{v^\tau_\itau}{U^{\tau-1}_\itau + \delta}\cdot x^\tau_{\itau}
  \\
  & = \frac1n \sumtau  \bigg(\frac{v^\tau_\itau}{U^{\tau-1}_\itau + \delta}\cdot x^\tau_{\itau}  + \sum_{i\neq \itau} \frac{\vtaui}{\Utaumi + \delta}\cdot \xtaui\bigg) 
  \\
  & = \frac1n \sumtau \sumiton \bigg( \frac{\vtaui}{\Utaumi + \delta}\cdot \xtaui\bigg) \;.
\$

\subsubsection{Step 2. Introducing a basic inequality.}

\begin{lemma}[\citet{bach2019universal}] For nonnegative numbers $ a_{1}, \ldots, a_{n} \in[0, a]$ and any $a_0\geq 0$, it holds
  \$\sum_{i=1}^{n} \frac{a_{i}}{a_{0}+\sum_{j=1}^{i-1} a_{j}} \leq 2+\frac{4 a}{a_{0}}+2 \log \left(1+\sum_{i=1}^{n-1} a_{i} / a_{0}\right)\; . \$
  \label{lm:bach}
\end{lemma}

We use the above inequality to bound the right-hand side of \eqref{eq:seeded bound by priced}. For a fixed agent $i$, let  $\vimax = \max_\tau \vtaui$, then
\$
& \sumtau \bigg( \frac{\vtaui}{\Utaumi + \delta}\cdot \xtaui\bigg) 
\\
& = \sumtau  \bigg( \frac{\vtaui \xtaui}{\sumstaui \vsi\xsi + \delta}\bigg) 
\\
& \leq 2 + 4\frac{\vimax}{\delta} + 2 \log \bigg(1+ \sum_{\tau = 1}^{T-1} \frac{\utaui}{\delta}\bigg) \tag{Invoking \Cref{lm:bach}}
\\
& = 2 + 4\frac{\vimax}{\delta} + 2 \log T + 2 \log \bigg(\frac1T + \frac1T \sum_{\tau = 1}^{T-1} \frac{\utaui}{\delta}\bigg)
\\
& \leq  2 + 4\frac{\vimax}{\delta} + 2 \log T + 2 \log \bigg(1 + \frac{\vmax}{\delta} \bigg) \;. 
\$
Finally, putting together we have the desired result. 
\$
& \frac1n \sumiton \frac{\widetilde{U}_i + \delta}{U_i + \delta} 
\\
& \leq 1 + \frac1n \sumiton \Bigg( 2 + 4\frac{\vimax}{\delta} + 2 \log T + 2 \log \bigg(1 + \frac{\vmax}{\delta} \bigg) \Bigg)
\\
& =  3 + \frac{4}{\delta} \bar v_{\mathrm{max}} + 2\log\bigg(1+\frac{\vmax}{\delta}\bigg) + 2 \log T \;.
\$

\section{Missing Proofs in \Cref{section-pace}}
\label{appendix-pace}

\subsection{Proof of \Cref{lemma-auxiliary-1}}

By the definition of $t^*$ and the bound on $\ell$ one can check:
\begin{equation*}
    U_1(A_1)\cdot S >\frac{1}{\varepsilon} \ell t^{*}. 
\end{equation*}
If $U_j^{s^\prime-1}(\bm v^\prime)< \ell (s^{\prime} -1)$, then 
\begin{equation*}
    \frac{\Hat{U}_j^{s^\prime-1}(\bm{v}^\prime)}{\Hat{U}_1^{s^\prime-1}(\bm{v}^\prime)} = \frac{\ell t^*}{U_1(A_1) \cdot S} <\varepsilon.
\end{equation*}
Now we show that the PACE allocation is preserved for the first $t^*$ items in the transformed sequence. For an original item $s\in A_j$ and corresponding item in $s^\prime$ in $\bm{v}^\prime$ ($s^\prime \leq t^*$), we check the two cases:
\begin{enumerate}
    \item  When $U_j^{s^\prime-1}(\bm v^\prime)\geq \ell (s^{\prime} -1)$, we have
    \begin{align*}
        \Hat{U}_1^{s^\prime-1}(\bm{v}^\prime) &= \max\{U_1(A_1), \ell (s^\prime-1)\} \geq \max\{U_1^s(\bm v), \ell (s-1)\} \geq  \Hat{U}_1^s(\bm v), \\
        \Hat{U}_j^{s^\prime-1}(\bm{v}^\prime) &= \max\{U_j^{s^\prime-1}(\bm v^\prime), \ell s^\prime\} = U_j^{s^\prime-1}(\bm v^\prime) = \max\{U_j^{s-1}(\bm v), \ell (s-1)\} = \Hat{U}_j^s(\bm v).
    \end{align*}
    Hence $\Hat{U}_j^{s^\prime-1}(\bm{v}^\prime)/\Hat{U}_1^{s^\prime-1}(\bm{v}^\prime) \geq \Hat{U}_j^s(\bm v)/\Hat{U}_1^s(\bm v) > v_j^s/v_1^s$. 
    \item When $U_j^{t^\prime-1}(\bm v^\prime)< \ell (t^{\prime} -1)$, by \Cref{lemma-auxiliary-1},
    \begin{align*}
         \frac{\Hat{U}_j^{s^\prime-1}(\bm{v}^\prime)}{\Hat{U}_1^{s^\prime-1}(\bm{v}^\prime)}<\varepsilon\leq \frac{v_j^{s}}{v_1^s}   
    \end{align*}
\end{enumerate}
Because we have set all agents' valuations to $0$ except agent $1$ and $j$, in both cases the item is indeed allocated to agent $j$, which is the same as the original input $\bm v$.

\subsection{Proof of \Cref{lemma-auxiliary-2}}
\label{proof-lemma-auxiliary-2}
In this proof, we let $a={U_1(A_1)}/{|A_1|} \in [\varepsilon, 1]$ be the per-round utility gain of agent $1$ in the beginning $S$ rounds.

To identify the sequence that gives maximum envy of agent $1$ on agent $j$'s bundle $B_j$, we use the same technique as the proof of \Cref{thm-envy}. The problem of maximizing agent $1$'s envy is equivalent to the following program:
\begin{equation}
\label{eq-optimization-envy-pace-base}
    \begin{aligned}
    \max_{v_1^{t_i}, v_j^{t_i}, t_i} & \ \ \frac{1}{aS}\sum_{i=1}^{k} v_1^{t_i} \\
    \text{s.t.} & \ \ S < t_1 < \cdots < t_k \leq t^* \\
                & \ \ v_1^{t_i}, v_j^{t_i} \in \{0\}\cup[\varepsilon, 1] \\
                & \ \ \frac{v_j^{t_i}}{v_1^{t_i}} \geq \frac{\Hat{U}_j^{t_{i-1}}}{\Hat{U}_1^{t_{i-1}}} \\
                & \ \ \Hat{U}_1^{t_i} = \max\{\ell {t_i}, aS\} \\
                & \ \ \Hat{U}_j^{t_i} = \max\{\ell {t_i}, {U}_j^{t_i}\} \\
                & \ \ {U}_j^{t_i} =  \sum_{i=1}^{k} v_{j}^{t_i}\\
    \end{aligned}
\end{equation}
$\{t_1, \cdots, t_k\}$ are $k$ positions where we put the items of $B_j$.  

Recall the technique we have used when proving \Cref{thm-envy}, where we introduced the function $q$ such that $q(u)$ is an upper bound of $v_1^{t+1}/v_j^{t+1}$ when $U_{j}^{t} = u$. When the horizon is infinitely long, each choice of $v_1^t$ makes up only a small step, and the envy is upper--bounded by the integral $\frac{1}{U_1^S} \int_{0}^{\infty} q(u) \mathrm{d} u$ in the asymptotic sense. 

There is a major difference when it comes to PACE. In the integral greedy algorithm, the upper bound of $v_1^{t+1}/v_j^{t+1}$ is clearly given by $\min\{\frac{1}{\varepsilon}, \frac{U_1^S}{u}\}$ once we know that $U_j^t = u$. However, in PACE, knowing only $U_j^t=u$ is not sufficient to bound on $v_1^{t+1}/v_2^{t+1}$. The bound may also depend on $t$, the position of the item.
\begin{equation*}
    \frac{v_1^{t}}{v_2^{t}} \leq \min\left\{\frac{1}{\varepsilon}, \frac{\max\{\ell (t-1), aS\}}{\max\{\ell (t-1), {U}_2^{t-1}\}}\right\}.
\end{equation*}
Thanks to \Cref{lemma-auxiliary-1} and the fact that $\ell t^{*}<aS$, we can simplify the above equation and remove the dependence on $t$.
\begin{equation}
    \label{eq-step-upper-bound}
    \frac{v_1^{t}}{v_2^{t}}  \leq
    \begin{cases}
        \frac{1}{\varepsilon} & U_2^{t-1}/aS \leq \varepsilon\\
        \frac{aS}{U_2^{t-1}} & \varepsilon< U_2^{t-1}/aS \leq 1/ \varepsilon \\
        0 & U_2^{t-1}/aS \geq 1/\varepsilon
    \end{cases}.
\end{equation}

The upper bound \eqref{eq-step-upper-bound} is identical to the constraint in a canonical optimization program of an envy maximizer parametrized by $aS$, defined as \eqref{eq-optimization-canonical} in the proof of \Cref{thm-envy}. Hence, both asymptotic upper bounds for envy and $|B_j \cap C_1|$ in the canonical optimization program still hold. 

\subsection{Proof of \Cref{lemma-infinite-utility}}
\label{proof-lemma-infinite-utility}
To prove \Cref{lemma-infinite-utility}, we show that for any given $c$, by choosing $\ell \leq \Bar{\ell}$ where $\Bar{\ell} <\frac{c\varepsilon}{1+(n-1)(1+\log 1/\varepsilon)}$, we have for any $i\in [n]$, 
$$\lim_{T\rightarrow \infty} \frac{U_i(A_i)}{T} \geq {\Bar{\ell}}.$$

We show this by contradiction. Assume there exists $i\in[n]$ such that $U_i(A_i) < T$. Let $C_i = \{t: v_i^t\}$ be the set of items which agent $i$ has non-zero value. Consider item $s$, the last item in $A_j\cap C_i$. To construct a contradiction, we show that $U_i(A_i)< \Bar{\ell} T$ implies an upper bound on each $|A_j\cap C_i|$, and thus $|C_i|$ (with a union bound). This might contradict the lower bound $V_i$. 

Consider an adversarial setting where the adversary attempts to maximize $|A_j \cap C_i|$, subject to the assumptions on the input and the allocation rule of PACE. 
Suppose $A_j \cap C_i=  \{t_1, \cdots, t_m\}$ where $t_1 < \cdots < t_m$. We relax the constraint on the adversary as follows:
\begin{enumerate}
    \item For $t\in A_j\cap C_i$, the adversary subject to constraint
    $$ \frac{v_j^t}{\Hat{U}_j^{t-1}} \geq \max_{k\in[n]}\frac{v_k^t}{\Hat{U}_k^{t-1}}. $$
    We relax this into
    \begin{equation}
    \label{eq-relaxed-step-1}
        \frac{v_j^t}{\Hat{U}_j^{t-1}} \geq \frac{v_i^t}{\Hat{U}_i^{t-1}}.
    \end{equation}
    \item Notice that $\Hat{U}_j^{t} \geq U_j^t$ and $\Hat{U}_i^{t} \leq \Bar{\ell} T$ by assumption. Then we can further relax \eqref{eq-relaxed-step-1} into 
    \begin{equation}
    \label{eq-relaxed-step-2}
        \frac{v_j^t}{{U}_j^{t-1}} \geq \frac{v_i^t}{\Bar{\ell} T}.
    \end{equation}
    \item The constraint on $U_j^{t}$ is minimum:
    \begin{equation}
        \label{eq-relaxed-step-3}
        U_j^t \geq \sum_{s\in A_j \cap C_i, s<t} v_j^s.
    \end{equation}
\end{enumerate}

The relaxed adversarial setting consists of constraint \eqref{eq-relaxed-step-2}, \eqref{eq-relaxed-step-3} and the value domain requirement $v_i^t, v_j^t \in [0,1]$. The value it concerns are $v_i^t, v_j^t$ where $t\in C_i \cap A_j$. 

Recall the definition of the canonical optimization program for an envy-maximizer in \eqref{eq-optimization-canonical}. We notice that the above relaxed constraint is identical to a canonical optimization program parametrized by $\Bar{\ell} T$. By \Cref{lemma-number-of-items}, the length of the value sequence in the relaxed setting is asymptotically bounded,
\begin{equation}
    \label{eq-relaxed-number-of-items-bound}
    \lim_{T\rightarrow \infty} \frac{|C_i\cap A_j|}{T} \leq \frac{\Bar{\ell}}{\varepsilon}\left(1+\log \frac{1}{\varepsilon}\right).
\end{equation}
Since \eqref{eq-relaxed-number-of-items-bound} is an upper bound for the relaxed setting, it must also hold for the true adversarial setting. We then know that 
\begin{align*}
    \lim_{T\rightarrow \infty} \frac{|C_i|}{T}
    &= \lim_{T\rightarrow \infty} \frac{|A_i|}{T} + \sum_{j\neq i} \lim_{T\rightarrow \infty} \frac{|C_i \cap A_j|}{T}\\
    &\leq \frac{\Bar{\ell}}{\varepsilon} \left(1+ (n-1)\left(1+\log \frac{1}{\varepsilon}\right)\right)\\
    &<c,
\end{align*}
where the last step is by the definition of $\Bar{\ell}$. 

However, this is impossible, since $|C_i| < cT$ implies $V_i < cT$. By contradiction we must have $U_i(A_i) \geq \Bar{\ell} T$ for sufficiently large $T$.

\subsection{Proof of \Cref{thm-cr-pace}}
\label{proof-thm-cr-pace}
Since $\ell < \frac{c\varepsilon}{1+(n-1)\left(1+\log 1/\varepsilon\right)}, r=1$, by \Cref{lemma-infinite-utility}, 
\begin{equation*}
    \bm{v} \in \mathcal{V}_{\varepsilon, c}^T \implies  \bm{v} \in \mathcal{V}_{\varepsilon}^T( \ell, r).
\end{equation*}
We can then apply the envy bound in \Cref{thm-envy-pace} to $\bm{v} \in \mathcal{V}_{\varepsilon, c}^T$ to get that for all $i\neq j$,
\begin{equation*}
    \lim_{T\rightarrow \infty} \frac{U_i(A_j)}{U_i(A_i)} \leq {1+2\log \frac{1}{\varepsilon}}.
\end{equation*}
Summing over all $j \in [n]$ and using $V_i = \sum_{j=1}^n U_i(A_j)$, we have
\begin{equation*}
    \lim_{T\rightarrow \infty} \frac{V_i}{U_i(A_i)} = \lim_{T\rightarrow \infty} \sum_{j=1}^n\frac{U_i(A_j)}{U_i(A_i)}  \leq 1+(n-1)\left({1+2\log \frac{1}{\varepsilon}}\right).
\end{equation*}
Applying the condition $V_i \geq cT$,
\begin{equation*}
\label{eq-pace-proportionality-pace}
    \lim_{T\rightarrow \infty} \frac{U_i(A_i)}{T} \geq c\left(1+(n-1)\left(1+2\log \frac{1}{\varepsilon}\right)\right)^{-1} \geq \frac{c}{n} \left(1+2\log \frac{1}{\varepsilon}\right)^{-1}.
\end{equation*}
On the other hand, by the AM-GM inequality, the optimal value of Nash Welfare can be bounded as follows:
\begin{equation*}
    \prod_{i=1}^n {U_i^\star} \leq \left(\frac{T}{n}\right)^n.
\end{equation*}
Following from the above two inequalities,
\begin{align*}
    \lim_{T\rightarrow \infty}\prod_{i=1}^n \frac{U_i^\star}{U_i(A_i)}\leq \lim_{T\rightarrow \infty} \prod_{i=1}^n U_i^\star \cdot \frac{T}{U_i(A_i)} \cdot T\leq \frac{1}{c} \left(1+2\log \frac{1}{\varepsilon}\right).
\end{align*}
This proves \Cref{thm-cr-pace}.

\end{document}